\documentclass[conference, 12pt, onecolumn]{ieeeconf}

\IEEEoverridecommandlockouts
\usepackage[usenames]{color}
\usepackage{enumerate}
\usepackage{url}
\usepackage{subfigure}
\usepackage{amsfonts,mathrsfs}
\usepackage{amssymb,amsmath}
\usepackage{verbatim}
\usepackage{acronym}
\usepackage{mathtools}
\usepackage{cite}
\usepackage{graphicx}
\usepackage{algorithm}
\usepackage[noend]{algpseudocode}

\def\fskip#1{}

\newtheorem{theorem}{Theorem}

\newtheorem{assumption}{Assumption}

\newtheorem{definition}{Definition}
\newtheorem{example}{Example}

\newtheorem{lemma}{Lemma}

\newtheorem{proposition}[theorem]{Proposition}
\newtheorem{remark}{Remark}

\renewcommand{\theenumi}{\arabic{enumi}}

\def\1{{\bf 1}}

\newcommand{\remove}[1]{}

\def\argmin{\mathop{\rm argmin}}
\def\argmax{\mathop{\rm argmax}}

\begin{document}
\title{Decentralized and Uncoordinated Learning of Stable Matchings: A Game-Theoretic Approach}
\author{\authorblockN{S. Rasoul Etesami$\dag$ and R. Srikant$\ddag$\vspace{-0.5cm}}
\thanks{$\dag$Department of Industrial and Systems Engineering, Department of Electrical and Computer Engineering, and Coordinated Science Laboratory, University of Illinois Urbana-Champaign,  Urbana, IL 61801 (Email: etesami1@illinois.edu).}
\thanks{$\ddag$Department of Electrical and Computer Engineering and Coordinated Science Laboratory, University of Illinois Urbana-Champaign,  Urbana, IL 61801 (Email: rsrikant@illinois.edu).}
\thanks{This research was supported by the AFOSR YIP FA9550-23-1-0107, AFOSR MURI FA9550-24-1-0002, NSF CAREER Award EPCN-1944403, and NSF CCF 22-07547.}
}
\maketitle
  

           
\begin{abstract}
We consider the problem of learning stable matchings with unknown preferences in a decentralized and uncoordinated manner, where ``decentralized" means that players make decisions individually without the influence of a central platform, and ``uncoordinated" means that players do not need to synchronize their decisions using pre-specified rules. First, we provide a game formulation for this problem with known preferences, where the set of pure Nash equilibria (NE) coincides with the set of stable matchings, and mixed NE can be rounded to a stable matching. Then, we show that for \emph{hierarchical} markets, applying the exponential weight (EXP) learning algorithm to the stable matching game achieves logarithmic regret in a fully decentralized and uncoordinated fashion. Moreover, we show that EXP converges locally and exponentially fast to a stable matching in general markets. We also introduce another decentralized and uncoordinated learning algorithm that globally converges to a stable matching with arbitrarily high probability. Finally, we provide stronger feedback conditions under which it is possible to drive the market faster toward an approximate stable matching. Our proposed game-theoretic framework bridges the discrete problem of learning stable matchings with the problem of learning NE in continuous-action games.
\end{abstract}

\section{Introduction}

Learning stable matchings is one of the fundamental problems in computer science, economics, and engineering that has received considerable attention over the past decades. Stable matchings provide a desirable notion of stability in two-sided matching markets where agents on each side of the market have preferences over the other side. A matching is called stable if no two agents prefer each other over their current matches. For instance, in college admissions, applicants have different preferences for colleges and vice versa. The goal is to match applicants to colleges in such a way that no applicant-college pair would prefer to break their current matches and instead be matched to each other. Similar situations arise in other applications, such as kidney exchange programs, job assignment to workers, matching in online dating platforms, and scheduling jobs on heterogeneous machines. 

It is known that when all preferences are known, stable matchings always exist, and a simple decentralized and uncoordinated \emph{Deferred-Acceptance} (DA) algorithm, first proposed by \cite{gale1962college}, can find such stable matches in a polynomial number of iterations. Here, a decentralized algorithm refers to one in which players make decisions individually without the influence of a central platform, and uncoordinated means that players do not need to synchronize their decisions using pre-specified rules. However, when preferences are unknown, developing such an algorithm faces major challenges due to the lack of coordination. While the DA algorithm provides a satisfactory solution for many practical applications, there are scenarios where either the information structure of the problem hinders the implementation of the DA algorithm or the stable matching achieved by the DA algorithm is suboptimal in terms of social welfare, especially when there are many feasible stable matchings in the market. For instance, there has been a recent emergence of online matching markets, such as online labor markets (e.g., TaskRabbit, Upwork), online dating markets (e.g., Tinder, Match.com), and online crowdsourcing platforms (e.g., Amazon Mechanical Turk), where users do not know their preferences a priori and can repeatedly interact with the market to improve their matching quality \cite{maheshwari2022decentralized}. The more a pair on both sides of the market gets matched, the more certain they become about their preferences. Therefore, an important question is how agents should interact with the market so that, in the absence of any coordination, they can learn their preferences as quickly as possible and achieve a stable matching through their collective behavior.

In this work, we consider a two-sided matching market framed as a marriage problem, where the agents on one side are referred to as men and the agents on the other side as women. We assume that both men and women have preferences regarding the other side, and that women are aware of their ordinal preferences for men. However, men do not know their preferences for women and only learn them if they propose and get matched. In this case, the matched men receive a noisy estimate of their preferences. We assume that agents cannot observe any other information (e.g., who is rejected or accepted) and cannot coordinate in any way. The goal for men is to follow a decentralized and uncoordinated proposal strategy such that the entire market converges to a stable matching over time.

One of the major challenges in learning a stable matching in markets with unknown preferences is handling collisions. More precisely, when multiple men propose to the same woman, only one of them gets matched and receives useful information, while all others are rejected and receive no information about their preferences for that woman. Therefore, resolving such collisions without coordination is a significant issue. Moreover, even if a man is matched with a woman, he receives a noisy utility drawn from an unknown distribution that may differ from his true preference. Consequently, a collision-avoidance process needs to be repeated many times before men can accurately estimate their true preferences. One of our goals in this work is to provide decentralized and uncoordinated algorithms that can effectively learn the true preferences while minimizing the number of collisions. Additionally, our work offers an alternative approach for measuring the closeness of the generated dynamics to any stable matching in general matching markets, rather than focusing on a specific stable matching in structured markets, as has been the main focus in past literature. To this end, we establish a connection between learning stable matchings in general markets with unknown preferences and learning Nash equilibria (NE) in noncooperative games. We leverage this connection to develop decentralized and uncoordinated algorithms for learning stable matchings in a principled manner.      

\subsection{Related Work}

Stable matching was first introduced by \cite{gale1962college} as a model for college admissions and to study the stability of marriages. Since then, there has been a tremendous effort to generalize and extend stable matchings to various market settings \cite{roth2008deferred}. In their seminal work, Gale and Shapley provided a simple Deferred Acceptance (DA) algorithm in which men propose to their most preferred women, and the women reject all proposals except the one from their most preferred man. They showed that such a decentralized algorithm converges to a stable matching in at most $O(n^2)$ steps, where $n$ is the total number of men and women in the market. Unfortunately, when the agents' preferences are unknown, the DA algorithm is no longer guaranteed to converge to a stable matching without additional coordination. More precisely, if men fail to propose to their highest preferred women due to a lack of knowledge of their true preferences, their subsequent proposals may become unpredictable. For instance, if a man $m$ is rejected by a woman $w$ at a certain time $t$, it does not necessarily mean that $m$ should stop proposing to $w$. It could be that another man mistakenly proposed to $w$ at time $t$, causing $m$'s rejection. Therefore, to extend the DA algorithm to the case of unknown preferences, some form of coordination among the agents is required to avoid uninformative acceptance/rejection patterns. To address this issue, earlier works \cite{pokharel2023converging, maheshwari2022decentralized, pagare2023two, kong2023player} have proposed various coordination-based DA algorithms. The main idea is to develop a phase-dependent process where, in some phases, the agents primarily explore, and in other phases, they implement the DA algorithm based on current estimates of their unknown preferences.

It is known that stable matchings with known preferences can be characterized using the extreme points of a fractional polytope that is totally dual integral \cite{vohra2012stable, teo1998geometry}. Moreover, stable matchings exhibit other useful properties that allow their geometry to be characterized using so-called \emph{rotations} \cite{kiraly2008total}. It was shown in \cite{roth1990random} that uncoordinated random better-response dynamics converge to a stable matching with probability one. Subsequently, \cite{ackermann2008uncoordinated} provided an exponential lower bound for the worst-case convergence time of the uncoordinated random better/best response dynamics to a stable matching. However, all these results apply to the case when market preferences are fully known, allowing agents to compute their best or better responses and update their decisions accordingly.

More recently, there has been a significant interest in matching markets with unknown preferences. However, depending on the information structure and the feedback received by the agents, one might expect a wide range of performance guarantees, which also depends on the type of algorithms followed by the agents (e.g., centralized/decentralized or coordinated/uncoordinated). The work in \cite{bei2013complexity} devised a randomized polynomial-time centralized algorithm for finding a stable matching with unknown deterministic preferences. Each time, the algorithm proposes a matching and receives feedback in the form of a blocking pair, where an unmatched pair of man and woman $(m, w)$ is called blocking if both $m$ and $w$ prefer each other over their current partners. The study in \cite{wang2022bandit} considers learning a specific stable matching with unknown stochastic preferences by adopting a suitable notion of stable regret, which measures the number of times that men propose to women other than their stable pair in that stable matching. The work in \cite{maheshwari2022decentralized} addressed the problem of learning stable matchings in \emph{hierarchically} structured markets where each submarket has a \emph{fixed pair} (i.e., a pair of men and women who prefer each other the most in that submarket). They developed a phase-coordinated decentralized algorithm that achieves logarithmic regret in time with respect to the market's unique stable matching. Additionally, \cite{liu2021bandit} considered learning stable matchings in an uncoordinated and decentralized fashion. However, they require stronger assumptions on information feedback (e.g., a player observes the actions of other players in the previous round) and use a different performance metric than the one we consider in this work. We refer to \cite{jagadeesan2021learning} and \cite{basu2021beyond} for other learning algorithms in two-sided matching markets with different performance guarantees. While all these works address the problem of learning a stable matching with unknown preferences, their algorithms differ substantially due to the information/feedback structure, the type of performance guarantee, and the level of coordination allowed among the agents.

\subsection{Contributions and Organization}

We consider the problem of learning stable matchings in matching markets where men do not know their preferences and aim to achieve a stable matching in a decentralized and uncoordinated manner. In particular, we focus on the weakest type of feedback that men can receive: a man observes a noisy version of his true preference only if his proposal is accepted, and receives no information otherwise. \emph{Our main objective is to provide a novel and principled approach to designing decentralized and uncoordinated learning algorithms for stable matchings by examining them through the lens of NE learning in noncooperative games.} Our contributions can be summarized as follows:
\begin{itemize}
    \item[(i)] We first provide a complete information continuous-action game formulation for the stable matching problem, namely the \emph{stable matching game}, and show that its set of pure NE coincides with the set of stable matchings. Additionally, its mixed NE points can be rounded in a decentralized way to obtain a stable matching. This connection provides an alternative approach for measuring the closeness of the market dynamics to stable matchings through the lens of NE computation in games, thus extending the conventional notion of \emph{regret} with respect to a specific stable matching in special markets to any stable matching in general markets.
    \item[(ii)] Leveraging this game-theoretic formulation, we present a simple decentralized and uncoordinated algorithm, EXP, that globally converges to a stable matching in hierarchical markets. This algorithm achieves logarithmic regret in time, thereby extending the existing phase-coordinated algorithms in the literature to an uncoordinated one. The propose EXP algorithm can be viewed as the exponential weight learning algorithm adapted to the stable matching game. However, unlike the conventional analysis of EXP in the multi-arm bandit setting with a constant stepsize \cite{bubeck2012regret}, the analysis of the EXP algorithm with time-varying parameters for the stable matching game is quite involved and requires a novel application of Freedman's concentration inequalities for martingales in a hierarchical fashion. 
    \item[(iii)] We then prove that the EXP algorithm always converges locally to a stable matching in general matching markets at an exponential rate. Moreover, we show that if the EXP algorithm converges globally with positive probability, it must converge to a stable matching.
    \item[(iv)] For general matching markets, it was previously unknown whether a fully decentralized and uncoordinated algorithm could globally learn a stable matching.  Therefore, we complement our results by providing an alternative decentralized and uncoordinated algorithm that globally converges to a stable matching in general markets with arbitrarily high probability (albeit with a slower convergence rate), thanks to the \emph{weakly acyclic} property of the stable matching game.
    \item[(v)] Finally, in Appendix \ref{appx:monotone-game}, we address the problem of designing information feedback in matching markets and leverage our game-theoretic framework to identify conditions that facilitate the rapid global learning of stable matchings in general markets.   
\end{itemize}

The paper is organized as follows. In Section \ref{sec:formulation}, we formally introduce the problem. In Section \ref{sec:game}, we provide a complete information game-theoretic formulation for the stable matching problem and prove several characterization results for its set of NE points. In Section \ref{sec:hierarchical}, we provide a decentralized and uncoordinated algorithm (EXP) that achieves logarithmic regret for hierarchical markets. In Section \ref{sec:dual-mirror}, we show that EXP  converges locally and exponentially fast to a stable matching in general matching markets. In Section \ref{sec:global}, we provide an alternative decentralized and uncoordinated algorithm that globally converges to a stable matching in general matching markets with arbitrarily high probability. We conclude the paper by identifying some future research directions in Section \ref{sec:conclusion}. Omitted proofs and other supplementary materials are given in Appendix \ref{appx:omitted}. 

\section{Problem Formulation}\label{sec:formulation}

\subsection{Stable Matchings with Known Preferences}

Here, we first introduce the stable matching problem with \emph{known} preferences \cite{gale1962college}. In this problem, there are a set $M$ of men and a set $W$ of women, where by introducing dummy agents with appropriate preferences, without loss of generality we may assume $|M|=|W|=n$ \cite{vohra2012stable} (we use the term ``agents" to refer to either men or women). Each man $m\in M$ has a cardinal preference for each women $w\in W$, denoted by $\mu_{mw}\in (0, 1]$, such if $\mu_{mw}>\mu_{mw'}$, it means that $m$ prefers $w$ over $w'$. Moreover, each woman $w$ has an ordinal preference over the men, and we write $m>_{w} m'$ if woman $w$ strictly prefers $m$ over $m'$. In this work we assume that no agent has ties in their preferences and we define 
\begin{align}\nonumber
\Delta=\min_{m, w\neq w'}|\mu_{mw}-\mu_{mw'}|,\ \ \ \ \mu_{\min}=\min_{m,w}\mu_{mw},\ \ \ \ \mu_{\max}=\max_{m,w}\mu_{mw}.
\end{align}
In particular, we note that $\Delta>0$ and $\mu_{\min}>0$.
\begin{definition}
Given a matching and two matched pairs $(m,w)$ and $(m',w')$, we say that $(m,w')$ forms a \emph{blocking pair} if $\mu_{mw'}>\mu_{mw}$ and $m>_{w'} m'$. In other words, $(m,w')$ is a blocking pair if both $m$ and $w'$ prefer each other over their current matches. A perfect matching is called a \emph{stable matching} if it does not contain any blocking pair. 
\end{definition}

\subsection{Stable Matchings with Unknown Preferences}

In this work, we consider the problem of finding a stable matching with the main difference that men do not know their true preferences $\mu_{mw}$, and they only get to learn them through interactions with the market. More precisely, for any $m\in M, w\in W$, we assume that the preference of man $m$ about woman $w$ is in the form of a $[0,1]$-supported unknown distribution $\mathcal{D}_{mw}$ with unknown mean $\mu_{mw}>0$. We assume men and women interact in this market through a discrete-time process, where at any time $t=0,1,2,\ldots$ that a man $m$ proposes to a woman $w$, he receives a feedback in the following form:
\begin{itemize}
    \item If the proposal of $m$ gets rejected at time $t$ because woman $w$ has received an offer from a more preferred man $m'$, i.e., $m'>_w m$, then $m$ receives no information as feedback other than the fact that he was rejected by woman $w$ at time $t$. 
    \item If the proposal of $m$ gets accepted by $w$ at time $t$, then $m$ observes an i.i.d. realization of his sampled preference drawn from $\mathcal{D}_{mw}$, denoted by $\hat{\mu}^t_{mw}$.    
\end{itemize}

In the stable matching problem with unknown preferences, the agents' goals are to interact with the market through the information feedback structure described above such that the emerging dynamics converge to a stable matching of the market with \emph{known} preferences $\{\mu_{mw}\}$. We note that in the discrete-time process described above, men are the decision makers on whom to propose at each time $t$;  women merely respond to men's decisions by accepting their most preferred proposal (if they received any) and rejecting all others. 
       
\section{A Game-Theoretic Formulation and Nash Equilibrium Characterization}\label{sec:game}

In this section, we provide a complete information noncooperative game formulation for the stable matching problem with known preferences whose set of pure Nash equilibrium (NE) points coincides with the set of stable matchings. Later, we will show how to leverage this formulation to extend our results for learning stable matchings with unknown preferences. Such a formulation has three main advantages: i) it reduces the learning task in matching markets to learning NE in continuous-action concave games, ii) it captures the selfish behavior of men (players) and the feedback they receive through their payoff functions, and iii) it simplifies the combinatorial structure inherent in learning stable matchings to learning NE in continuous-action games. While analyzing the agents' decisions directly on the discrete space of matchings seems complicated, which is why some earlier works need to restrict attention to special market structures \cite{maheshwari2022decentralized}, our game-theoretic reduction simplifies this task to some extent, thanks to the extensive literature on NE learning in continuous-action games.

\subsection{Stable Matching Game}
Consider a complete information noncooperative game in which the action set of each man (player) is given by the set of women $W$, and the set of (mixed) strategies for each player is given by the probability simplex over the set of women $W$, i.e., the strategy set of player $m$ is defined by $$\mathcal{X}_m=\Big\{x_m\in \mathbb{R}_{+}^{|W|}: \sum_{w\in W}x_{mw}= 1 \Big\}.$$ 
Let $x_{-m}$ denote the strategy vector of all the players except the $m$th one. Given a strategy profile $x=(x_m,x_{-m})$, we define the payoff of player $m$ by
\begin{align}\label{eq:game-payoff}
u_m(x)=\sum_{w\in W}\Big(\mu_{mw}\prod_{k>_w m}(1-x_{kw})\Big)x_{mw},
\end{align}
where $\mu_{mw}>0$ is the utility (true preference) received by man $m$ if he gets matched to women $w$.

\begin{definition}
A strategy $x_m$ for player $m$ is called \emph{pure} if it is zero in all coordinates except one. Otherwise, it is called a mixed strategy. A pure (mixed) strategy profile $x^*$ is called a pure (mixed) NE if for any player $m$ and any pure (mixed) strategy $x_m$, we have $u_m(x_m,x^*_{-m})\leq u_m(x^*_m,x^*_{-m})$. 
\end{definition}

\begin{remark}\label{rem:bipartitie}
A strategy $x$ can be represented as a weighted bipartite graph $\mathcal{H}(x)=(M\cup W, E(x))$ between men and women. In this graph, the edges $E(x)=\{(m,w): x_{mw}>0\}$ are supported over the nonzero entries of $x$, with the weight of edge $(m,w)$ given by $x_{mw}$.
\end{remark}

One can interpret $u_m(x)$ as the expected utility that player $m$ would receive by successfully getting matched if each man independently proposes to women according to his mixed strategy distribution $x_m$. The reason for defining the payoff functions as in \eqref{eq:game-payoff} is that we want our devised algorithms to be implemented in a fully decentralized and uncoordinated manner among men by relying only on their received feedback (embedded into their payoff functions). While such payoff functions are highly nonlinear, they are essential to eliminate any degrees of coordination among the players. In other words, the cost of devising a fully coordination-free algorithm comes in analyzing more complex payoff functions with higher degrees of nonlinearity.\footnote{We refer to Appendix \ref{appx:omitted} for a simpler form of payoff functions with less degree of nonlinearities, which in turn requires stronger information feedback.} The following are three properties of the payoff functions given in \eqref{eq:game-payoff}. 
\begin{itemize}
\item[(i)] Player $m$'s strategy does not have any impact on the coefficient terms $\mu_{mw}\prod_{k>_w m}(1-x_{kw})$, $ w\in W$. In particular, the gradient of $u_m(\cdot)$ with respect $x_{mw}$ is given by $$v_{mw}(x):=\nabla_{mw}u_{m}(x)=\mu_{mw}\prod_{k>_w m}(1-x_{kw}).$$ 
\item[(ii)] The payoff of each player $m$ is linear with respect to his own strategy $x_m$.
\item[(iii)] For any fixed strategy of other players $x_{-m}$, player $m$ always has a best response among pure strategies that is obtained by setting $x_{mw}=1$ for the woman $w$ that achieves the maximum value $v_{mw}(x)$, and $x_{mw}=0$ otherwise.    
\end{itemize}

In the remainder of this paper, we will refer to the above noncooperative game as the \emph{stable matching game} and denote it by $\mathcal{G}=(M, \{u_m\}_{m\in M}, \{\mathcal{X}_m\}_{m\in M})$. It is important to note that although our interest is in obtaining a NE of the stable matching game, which is a \emph{complete information} one-shot game, our goal is to learn such a NE by repeatedly playing the \emph{incomplete information} game, where the true preferences $\mu_{mw}$ (and hence the payoff functions) are not fully known. 

\subsection{Nash Equilibrium Characterization}
In this part, we show that the set of pure and mixed NE points of the stable matching game has interesting connections with the set of stable matchings when preferences are known. The following theorem establishes one such result. 
 
\smallskip 
\begin{theorem}\label{thm:NE-Stable}
A pure strategy profile $x^*\in \{0,1\}^{n^2}$ corresponds to the characteristic vector of a stable matching if and only if it is a pure NE for the stable matching game.
\end{theorem}
\begin{proof}
Suppose $x^*$ is a pure NE. Then, we cannot have $x^*_{mw}=x^*_{m'w}=1$ for some woman $w$ and $m\neq m'$. Otherwise, between $m$ and $m'$, the one that woman $w$ prefers less receives zero utility, while he can always deviate to propose to a woman $w'$ that does not have any offer and strictly improve his utility. Therefore, $x^*$ must correspond to the characteristic vector of a perfect matching. To show that $x^*$ is a stable matching, let us, by contradiction, assume that there exists a blocking pair with respect to the matching $x^*$ and denote it by $(m,w)$. Moreover, let $m'$ and $w'$ be the corresponding matches for $w$ and $m$ under $x^*$, that is $x^*_{mw'}=1$ and $x^*_{m'w}=1$. Since $(m,w)$ is a blocking pair, we have $\mu_{mw}>\mu_{mw'}$ and $m>_w m'$. The payoff of player $m$ under NE $x^*$ equals $u_m(x^*)=\mu_{mw'}$. Now if player $m$ changes its strategy to $x_m:=(\boldsymbol{0},x_{mw}=1)$, his payoff becomes 
\begin{align}\nonumber
u_m(x_m,x^*_{-m})=\sum_{j\in W}\mu_{mj}\prod_{k>_j m}(1-x^*_{kj})x_{mj}=\mu_{mw}\prod_{k>_w m}(1-x^*_{kj})=\mu_{mw},
\end{align}
where the second equality holds because $m'<_{w}m$ so that the term $1-x^*_{m'w}$ does not appear in the product $\prod_{k>_w m}(1-x^*_{kw})$, and moreover $x^*_{kw}=0\ \forall k\neq m'$, which implies $\prod_{k>_w m}(1-x^*_{kw})=1$. This shows that $u_m(x_m,x^*_{-m})=\mu_{mw}>\mu_{mw'}=u_m(x^*)$, contradicting the fact that $x^*$ is a NE.  

Conversely, let $x^*$ be the characteristic vector of a stable matching, and by contradiction, assume it is not a NE. Then, there exists at least one player $m$ who can deviate and strictly improve his payoff. In particular, we can consider his pure strategy best response move, which is guaranteed by (iii), and denote his most preferred woman by $w$. Moreover, let $m'$ and $w'$ be the corresponding matches for $w$ and $m$ under the stable matching $x^*$, i.e., $x^*_{m'w}=x^*_{mw'}=1$. The current payoff of $m$ under the pure strategy $x^*$ equals $u_m(x^*)=\mu_{mw'}$, while his payoff after such deviation becomes
\begin{align}\label{eq:NE-converse}
u_m(x_m,x^*_{-m})=\sum_{j\in W}\mu_{mj}\prod_{k>_j m}(1-x^*_{kj})x_{mj}=\mu_{mw}\prod_{k>_w m}(1-x^*_{kw})\leq \mu_{mw}.
\end{align} 
Since $0< \mu_{mw'}=u_m(x^*)<u_m(x_m,x^*_{-m})$, we obtain $\mu_{mw'}< \mu_{mw}$. Moreover, from relation \eqref{eq:NE-converse} and the fact that $u_m(x_m,x^*_{-m})>0$, we must have $\prod_{k>_w m}(1-x^*_{kw})=1$, which implies $m'<_w m$. This together with $\mu_{mw'}< \mu_{mw}$ shows that $(m,w)$ is a blocking pair, contradicting that $x^*$ is stable.
\end{proof}

While the above equilibrium characterization theorem concerns pure NE points, it is possible that the stable matching game admits a mixed NE. The following is one example.

\smallskip 
\begin{example}\label{ex:mixed}
Let us consider an instance of the stable matching game with three men $M=\{m_1,m_2,m_3\}$ and three women $W=\{w_1,w_2,w_3\}$. We assume the preferences of men and women are given by
\begin{align}\nonumber 
&\mu_{m_1w_1}=2,\ \mu_{m_1w_2}=1,\ \mu_{m_1w_3}=\epsilon, \ \ \ \ \ \ \ \ \ \ \ m_3<_{w_1} m_1<_{w_1} m_2,\cr
&\mu_{m_2w_1}=3,\ \mu_{m_2w_2}=4,\ \mu_{m_2w_3}=5, \ \ \ \ \ \ \ \ \ \ \ m_3<_{w_2} m_2<_{w_2} m_1,\cr
&\mu_{m_3w_1}=\epsilon,\ \mu_{m_3w_2}=12,\ \mu_{m_3w_3}=6,\ \ \ \ \ \ \ \ \ \ m_1<_{w_3} m_2<_{w_3} m_3,
\end{align}
where $\epsilon>0$ is a small positive constant. Then, the stable matching game has a mixed NE that is supported over seven edges comprising of two cycles of length 4 that share the common edge $(m_2,w_2)$ (see, Remark \ref{rem:bipartitie}). In particular, the seven positive entries of that mixed NE are given by 
\begin{align}\nonumber
&x_{m_1w_1}=\frac{\mu_{m_2w_1}}{\mu_{m_2w_2}},\ \ \ \ \ \ \ \ \ \ \ \ \ \ x_{m_1w_2}=1-\frac{\mu_{m_2w_1}}{\mu_{m_2w_2}},\cr
&x_{m_2w_1}=1-\frac{\mu_{m_1w_2}}{\mu_{m_1w_1}}, \ \ \ \ \ \ \ \ \ x_{m_2w_2}=1-\frac{\mu_{m_2w_2}\mu_{m_3w_3}}{\mu_{m_2w_1}\mu_{m_3w_2}}, \ \ \ \ \ \ x_{m_2w_3}=\frac{\mu_{m_1w_2}}{\mu_{m_1w_1}}+\frac{\mu_{m_2w_2}\mu_{m_3w_3}}{\mu_{m_2w_1}\mu_{m_3w_2}}-1,\cr
&x_{m_3w_2}=\frac{\mu_{m_2w_1}}{\mu_{m_2w_3}}, \ \ \ \ \ \ \ \ \ \ \ \ \ \ x_{m_3w_3}=1-\frac{\mu_{m_2w_1}}{\mu_{m_2w_3}}.
\end{align}   
\end{example}

\medskip
As Example \ref{ex:mixed} suggests, mixed NE points may correspond to unpredictable bipartite graphs (Remark \ref{rem:bipartitie}), making it difficult to make general statements about their structure. However, the following theorem establishes an interesting connection between mixed and pure NE points, implying that obtaining a mixed NE with full support on the women's side is as satisfactory as obtaining a pure NE. 

\smallskip
\begin{theorem}\label{thm:mixed-pure}
Let $x$ be any mixed NE in which each woman receives at least one (fractional) proposal. For each man $m$, let $w_m$ be the least preferred woman among those that he fractionally proposes to them, i.e., $w_m=\argmin_{w\in W} \{\mu_{mw}: x_{mw}>0\}$. Then $\{(m,w_m), m\in M\}$ forms a stable matching (and hence a pure NE) for the stable matching game. 
\end{theorem}
\begin{proof}
Given a mixed NE $x$, let $\mathcal{H}(x)=(M\cup W, E(x))$ be the induced bipartite graph by $x$, where $E(x)=\{(m,w): x_{mw}>0\}$ denotes all the edges with fractionally positive proposals. Moreover, let us denote the neighboring nodes of $m$ and $w$ in $\mathcal{H}(x)$ by $N_m(x)$ and $N_w(x)$, respectively. Since $x$ is a NE, the value $v_{mw}(x)=\mu_{mw}\prod_{k>_w m}(1-x_{kw})$ should be the same for all $w\in N_m(x)$ and equals to $u_m(x)$. Otherwise, using properties (i) and (ii), player $m$ could strictly increase his payoff by putting his entire strategy mass on the woman $w^*=\argmax_{w\in N_m(x)}v_{mw}(x)$. The fact that $u_m(x)=v_{mw}(x)\ \forall w\in N_m(x)$ follows because $x_m$ is a probability distribution that is supported over $N_m(x)$.

Next, we note that a man $m$ cannot be the most preferred choice for more than one woman in $\mathcal{H}(x)$. Otherwise, let $w_1,w_2\in N_m(x)$ be two different women who rank $m$ the highest among their fractional matches in $N_{w_1}(x)$ and $N_{w_2}(x)$, respectively. Then, we have $$\prod_{k>_{w_1} m}(1-x_{kw_1})= \prod_{k>_{w_2} m}(1-x_{kw_2})=1,$$ which together with the above argument implies that
\begin{align}\nonumber
\mu_{mw_1}=\mu_{mw_1}\prod_{k>_{w_1} m}(1-x_{kw_1})=v_{mw_1}(x)=u_m(x)=v_{mw_2}(x)=\mu_{mw_2}\prod_{k>_{w_2} m}(1-x_{kw_2})=\mu_{mw_2}.
\end{align}  
This contradicts the fact that there are no ties between the preferences of man $m$.    

Since each woman receives at least one fractional proposal, $N_w(x)\neq \emptyset \ \forall w$. Let $m_w$ be the unique most favorite man of $w$ in $N_w(x)$. Using the above argument, the elements of $\{m_w, w\in W\}$ must be distinct such that $\mathcal{M}_1=\{(m_w,w), w\in W\}$ forms a matching. Viewing this matching from men's side, one can represent $\mathcal{M}_1$ in the form of $\mathcal{M}_2=\{(m,w_m), m\in M\}$, where $w_m=\argmin_{w\in N_m(x)} \mu_{mw}$ is the least preferred woman among those that man $m$ fractionally proposed to. To see this equivalence, fix a man $m$ and let $w_0$ be his match in $M_1$, i.e., $m_{w_0}=m$. Observe that each $w'\in N_m(x)\setminus \{w_0\}$ has a unique distinct most favorite man $m_{w'}\neq m$. That means that 
\begin{align}\nonumber
v_{mw'}(x)=\mu_{mw'}\prod_{k>_{w'} m}(1-x_{kw'})\leq \mu_{mw'}(1-x_{m_{w'}w'})<\mu_{mw'}. 
\end{align}
On the other hand, since $w_0$ ranks $m$ the highest among those who propose to her, we have $v_{mw_0}(x)=\mu_{mw_0}\prod_{k>_{w_0} m}(1-x_{kw'})=\mu_{mw_0}$. Since at a NE $v_{mw_0}(x)=v_{mw'}(x)$, we must have $\mu_{mw_0}<\mu_{mw'}, \forall w'\in N_m(x)\setminus \{w_0\}$, which implies $w_0=w_m$. Thus, for any pair $(m_{w_0},w_0)\in \mathcal{M}_1$ there is a pair $(m,w_m)\in M_2$ such that $(m_{w_0},w_0)=(m,w_m)$, which shows that $\mathcal{M}_1=\mathcal{M}_2$. 

Finally, we will show that the binary characteristic vector of the matching $\{(m,w_m), m\in M\}$, denoted by $\hat{x}$, is indeed a pure NE. We will show that by arguing that any strictly better deviation in $\hat{x}$ is also a strictly better deviation in $x$, contradicting the fact that $x$ is a NE. On the contrary, let $m_w$ be a man who wants to deviate from his pair $w$ to a woman $w'\neq w$. Such a deviation is beneficial to player $m_{w}$ only if $m_{w'}<_{w'} m_w$. Otherwise, his payoff will reduce from $u_{m_w}(\hat{x})=\mu_{m_ww}$ to $0$. Henceforth, we may assume $m_{w'}<_{w'} m_w$. The fact that the payoff of player $m_w$ is currently $u_{m_w}(\hat{x})=\mu_{m_ww}$ and he wants to deviate to $w'$ implies that $\mu_{m_ww}<\mu_{m_ww'}$. As $m_w$ is the most preferred man for woman $w$ in $\mathcal{H}(x)$, we also have $u_{m_w}(x)=\mu_{m_ww}$. Since $m_{w'}<_{w'} m_w$, the woman $w'$ would prefer $m_w$ over all the men in $N_w'(x)$. Thus, in the mixed NE $x$, man $m_w$ could also deviate by solely proposing to $w'$, in which case his payoff would strictly increase from $u_{m_w}(x)=\mu_{m_w,w}$ to $\mu_{m_ww'}$, which is a contradiction.
\end{proof}

\begin{remark}
As can be seen from the proof of Theorem \ref{thm:mixed-pure}, if the condition that each woman receives at least one proposal does not hold, the theorem still applies to the induced submarket in which each woman has received at least one fractional proposal.
\end{remark}

Theorem \ref{thm:mixed-pure} suggests some strategies for obtaining decentralized solutions. We can first design a decentralized algorithm (see Algorithm \ref{alg:main} in conjunction with the preference estimation oracle of Algorithm \ref{alg:gap-estimation}) to learn a mixed NE of the stable matching game and acquire estimated preferences. Subsequently, each man can convert his fractional mixed strategy into a pure strategy in a decentralized manner by proposing only to his least (estimated) preferred woman among those he fractionally proposes to. This process ensures that the resulting pure strategy profile will be a pure NE with high probability, and the probability of error vanishes over time. Importantly, each man has complete access to his own estimated preferences, and the rounding process outlined in Theorem \ref{thm:mixed-pure} does not require knowledge of others' mixed strategies. Each man only needs to be aware of his own mixed strategy and estimated preferences. This approach offers a simple decentralized method to convert mixed NE points into stable matchings. 

\section{Logarithmic Regret for Hierarchical Markets}\label{sec:hierarchical}

In this section, we consider the stable matching game with unknown preferences and develop a decentralized and uncoordinated algorithm for learning its NE points. In particular, we focus on matching markets with specific structures because without imposing any assumption on the matching market, there are exponential lower bounds for the number of iterations to find a stable matching through better response dynamics \cite{ackermann2008uncoordinated}. One example of such a structured matching market is the \emph{hierarchical} matching market, where it was shown in \cite{maheshwari2022decentralized} that if men follow a certain decentralized (but coordinated) algorithm, the men's expected regret (see Definition \ref{def:regret}) is at most logarithmic in time but exponential in terms of other parameters such as the number of men $n$. In this section, we show that a much simpler algorithm (Algorithm \ref{alg:main}) also achieves logarithmic regret in a fully uncoordinated fashion, hence, removing the phase-dependent coordination required by the algorithm in \cite{maheshwari2022decentralized}. In particular, our analysis provides a cleaner characterization of the other constants involved in the regret bound. Throughout this section, we impose the following hierarchical assumption on the matching market.

\begin{assumption}\label{ass:hirarchy}
We assume that the sets of men and women can each be ordered as $M=\{m_1,\ldots,m_n\}$ and $W=\{w_1,\ldots,w_n\}$ such that any man and woman with the same rank prefer each other above any other partner with a lower rank; i.e., man $m_k$ prefers woman $w_k$ to women $w_{k+1}, w_{k+2},\ldots,w_{n}$, and woman $w_k$
prefers man $m_k$ to men $m_{k+1}, m_{k+2},\ldots,m_n$. It is easy to see that under this hierarchical assumption, the matching $\{(m_k,w_k), k=1,\ldots,n\}$ is a stable matching. 
\end{assumption} 

\begin{remark}
The condition imposed in Assumption \ref{ass:hirarchy} is also known as the \emph{Sequential Preference Condition} (SPC) \cite{clark2006uniqueness} and is weaker than the $\alpha$-reducible condition considered in \cite{maheshwari2022decentralized}, which assumes that each submarket has a \emph{fixed-pair}: A pair $(m,w)$ is called a fixed pair if both $m$ and $w$ prefer each other the most among anyone else in that submarket. In fact, one can show that a matching market has a unique stable matching if and only if it satisfies $\alpha$-reducible condition \cite{maheshwari2022decentralized}. Since Assumption \ref{ass:hirarchy} is weaker than the $\alpha$-reducible condition, all our results immediately applies to $\alpha$-reducible markets. 
\end{remark}

\begin{definition}\label{def:regret}
Let us assume that players follow a decentralized and uncoordinated algorithm $\mathcal{A}$ that results in a sequence of proposals $\alpha^t=(\alpha^t_m, m\in M), t=1,\ldots,T$. The \emph{expected regret} of algorithm $\mathcal{A}$ is defined by
\begin{align}\nonumber
R(T)=\sum_{t=1}^T \mathbb{E}\big[\boldsymbol{1}_{\{\exists k\in [n]: \alpha^{t}_{m_k}\neq w_k\}}\big],  
\end{align}
where $\boldsymbol{1}_{\{\cdot\}}$ denotes the indicator function. In other words, $R(T)$ is the expected number of times that the men's proposals do not form a stable matching. 
\end{definition}

\subsection{Algorithm Design and Preliminaries}

The algorithm that we propose is an adaptation of the online dual mirror descent for adversarial multi-arm bandit problem (see, e.g., EXP3 Algorithm in \cite[Chapter 11]{lattimore2020bandit}) to the stable matching game, where for simplicity we refer to it as the EXP algorithm. In particular, it uses an entropy regularizer and adaptive stepsize for updating the players' mixed strategy distributions over time.

To describe the EXP algorithm, let us denote the gradient vector of player $m$ by $v_m(x)=\nabla_m u_m(x)$ such that for any woman $w$, we have $v_{mw}(x)=\mu_{mw}\prod_{k>_w m}(1-x_{kw})$. In particular, we note that $u_m(x)=\langle x_m, v_m(x)\rangle$. Therefore, when restricting to a pure strategy (action) profile $\alpha^t=(\alpha_1^t,\ldots,\alpha^t_n)\in W\times\cdots\times W$ that is played at time $t$, the $w$-th coordinate of the gradient vector is given by
\begin{align}\label{eq:v_alpha}
v_{mw}^t:=\mu_{mw}\boldsymbol{1}_{\{\alpha^t_k\neq w\ \forall k>_{w} m\}},
\end{align}
where $\boldsymbol{1}_{\{\cdot\}}$ denotes the indicator function. On the other hand, by playing a pure strategy profile $\alpha^t$, each man $m$ receives a feedback only for the woman $\alpha^t_m$ that he has proposed to her at time $t$. In particular, if $\alpha_m^t=w$, player $m$ observes the payoff $\mu^t_{mw}\boldsymbol{1}_{\{\alpha^t_k\neq w \ \forall k>_w m\}}\boldsymbol{1}_{\{\alpha_m^t=w\}}$ at time $t$, where we recall that $\mu^t_{mw}$ is the realization of a random reward drawn independently from an unknown distribution $\mathcal{D}_{mw}$ with unknown mean $\mu_{mw}$ and bounded variance $\sigma^2_{mw}\in [0, 1]$.

During the course of the algorithm, at each time $t$, each player $m$ holds a mixed strategy $X^t_m\in \mathcal{X}_m$ and updates it whenever he receives new information. He then computes his adjusted mixed strategy $\hat{X}^{t}_{m}=(1-\gamma^t)X^t_{m}+\frac{\gamma^t}{n}\boldsymbol{1}$, where $\gamma^t>0$ is a mixing parameter and $\boldsymbol{1}$ is the vector of all ones. Using ideas from importance sampling for bandit problems \cite{bubeck2012regret,giannou2021rate}, and to obtain an unbiased estimator of the actual gradient vector, player $m$ constructs an estimate vector $\hat{v}^t_m$ by normalizing the received feedback with his adjusted mixed strategy probabilities as
\begin{align}\label{eq:noise-feedback}
\hat{v}^{t}_{mw}=\frac{\mu^t_{mw}\boldsymbol{1}_{\{\alpha^t_k\neq w\ \forall k>_{w} m\}}}{\hat{X}^t_{mw}}\cdot\boldsymbol{1}_{\{\alpha_m^t=w\}}, \ \ \ \ w\in W.
\end{align}
The reason why we use adjusted mixed strategies $\hat{X}^t$ with additional mixing parameter $\gamma^t$ is to ensure that the mixed strategies remain bounded away from zero by a positive quantity. That makes the estimators \eqref{eq:noise-feedback} have bounded variance (see Lemma \ref{lemm:unbiased-martingale}), which allows us to use martingale concentration results to bound the probability of various events.

Let us denote the (random) score vector of player $m$ at time $t-1$ by $\hat{L}^{t-1}_m=(\hat{L}_{mw},w\in W)$. Player $m$ uses $\hat{L}^{t-1}_m$ to compute his  mixed strategy vector $X_m^t$ using the logit update rule \eqref{eq:logit}. He then proposes to a woman $\alpha_{m}^{t}$ chosen independently at random according to his adjusted mixed strategy $\hat{X}_m^t$, and receives as feedback $\hat{v}_{m}^{t}$ with entries given by \eqref{eq:noise-feedback}. Player $m$ then updates his score vector by $\hat{L}^t_{m}=\hat{L}^{t-1}_{m}+\hat{v}^t_{m}$ and proceeds to the next round. The detailed description of the EXP algorithm is summarized in Algorithm \ref{alg:main}.

\begin{algorithm}[t]\caption{EXP: An Uncoordinated and Decentralized Algorithm for Player $m$}\label{alg:main}
{\bf Input:} A decreasing stepsize sequence $\{\eta^t\}$, a mixing sequence $\{\gamma^t\}$, and an initial vector $\hat{L}^0_m=\boldsymbol{0}$.
\noindent
{\bf For} $t=1,2,\ldots$, player $m$ independently performs the following steps:
\begin{itemize}
\item Player $m$ computes his mixed-strategy vector $X_m^t$ using the logit map
\begin{align}\label{eq:logit}
X^{t}_{mw}=\frac{\exp(\eta^t \hat{L}^{t-1}_{mw})}{\sum_{w'}\exp(\eta^t \hat{L}^{t-1}_{mw'})}, \ \ w\in W.
\end{align}
\item Player $m$ draws a pure strategy $\alpha^{t}_m\in W$ according to his adjusted mixed strategy defined by
\begin{align}\nonumber
\hat{X}^{t}_{m}=(1-\gamma^t)X^t_{m}+\frac{\gamma^t}{n}\boldsymbol{1},\ \ w\in W,
\end{align}
where $\gamma^t>0$ is a mixing parameter and $\eta^t>0$ is the stepsize. 
\item Player $m$ receives a feedback in terms of his observed payoff at time $t$ and constructs an unbiased estimate of his actual payoff gradient vector:
\begin{align}\label{eq:estimators}
&\hat{v}^{t}_{mw}=\frac{\mu^t_{mw}\boldsymbol{1}_{\{\alpha^t_k\neq w\ \forall k>_{w} m\}}}{\hat{X}^t_{mw}}\cdot \boldsymbol{1}_{\{\alpha_m^t=w\}},\ \ w\in W,
\end{align}
where $\mu^t_{mw}$ is the realization of the random reward if $m$ is matched to $w$ at time $t$, and updates
\begin{align}\nonumber
\hat{L}^t_{mw}=\hat{L}^{t-1}_{mw}+\hat{v}^t_{mw},\ \ w\in W.
\end{align}
\end{itemize}
\end{algorithm}

In order to analyze the convergence behavior of Algorithm \ref{alg:main} to a stable matching, let $\{\mathcal{F}^{t-1}\}_{t=1}^{\infty}$ be the increasing filtration sequence that is adapted to the history of the random processes generated by Algorithm \ref{alg:main}. More precisely, $$\mathcal{F}^{t-1}=\{X^{\tau},\hat{L}^{\tau},\alpha^{\tau},\hat{v}^{\tau}, \tau=0,1,\ldots,t-1\}\cup\{X^t\},$$ contains all the realized events up to time $t-1$ except the realization of pure strategies $\alpha^t$ and the rewards $\mu^t$ at time $t$. In particular, all relevant processes at time $t-1$ as well as $X^t$ and $\hat{X}^t$ are $\mathcal{F}^{t-1}$-measurable, but $\alpha^{t}$ and $\hat{v}^{t}$ are not $\mathcal{F}^{t-1}$-measurable. 

We begin with the following lemma, which shows that the feedback $\hat{v}^{t}_m$ received by player $m$ at time $t$ is conditionally an unbiased and bounded estimator of the actual gradient vector $v_m^t$. The proof of this lemma follows from standard analysis of the multi-arm bandit setting \cite{bubeck2012regret}. However, for the sake of completeness, we provide a short proof here.

\begin{lemma}\label{lemm:unbiased-martingale}
The received feedback is a conditionally unbiassed and bounded estimate of the gradients: 
\begin{align}\nonumber
\mathbb{E}[\hat{v}_{m}^{t}|\mathcal{F}^{t-1}]=v_m(\hat{X}^{t})\ \ \forall m, \ \ \ \ \ \mathbb{E}[(\hat{v}_{mw}^{t})^2|\mathcal{F}^{t-1}]\leq \frac{\sigma^2}{\hat{X}^t_{mw}}\leq \frac{n}{\gamma^t}\ \ \forall m,w,
\end{align}
where $\sigma^2:=\max_{m,w}\mathbb{E}[(\hat{\mu}^t_{mw})^2]\in [0, 1]$.
\end{lemma}
\begin{proof}
Consider any pair $(m,w)$. Since $\hat{X}^t_{mw}$ is $\mathcal{F}^{t-1}$-measurable and $\hat{\mu}^t_{mw}$ is drawn independently from the algorithm's past decisions (and hence is independent of $\mathcal{F}^{t-1}$ and $\alpha^t$), we have
\begin{align}\nonumber
\mathbb{E}[\hat{v}_{mw}^{t}|\mathcal{F}^{t-1}]&=\frac{\mathbb{E}[\mu^t_{mw}]}{\hat{X}^t_{mw}}\cdot \mathbb{E}[\boldsymbol{1}_{\{\alpha^t_k\neq w \ \forall k>_w m\}}\cdot \boldsymbol{1}_{\{\alpha_m^t=w\}}|\mathcal{F}^{t-1}]\cr 
&=\frac{\mu_{mw}}{\hat{X}^t_{mw}}\mathbb{P}\{\alpha_m^t=w|\mathcal{F}^{t-1}\}\prod_{k>_w m}\mathbb{P}\{\alpha^t_k\neq w|\mathcal{F}^{t-1}\}\cr
&=\mu_{mw}\prod_{k>_w m}(1-\hat{X}^t_{kw})=v_{mw}(\hat{X}^{t}),
\end{align}
where the second equality holds because conditioned on $\mathcal{F}^{t-1}$, the random variables $\{\alpha^t_k, \forall k\}$ are mutually independent as they are drawn independently from the given distributions $\{\hat{X}^t_k, \forall k\}$. Similarly, we can bound the expected second moment as
\begin{align}\nonumber
\mathbb{E}\big[\big(\hat{v}_{mw}^{t}\big)^2|\mathcal{F}^{t-1}\big]&=\frac{\mathbb{E}[(\hat{\mu}^t_{mw})^2]}{\hat{X}^t_{mw}} \prod_{k>_w m}\mathbb{P}\{\alpha^t_k\neq w|\mathcal{F}^{t-1}\}\leq \frac{\sigma^2}{\hat{X}^t_{mw}}\leq \frac{n}{\gamma^t}, \ \ \forall m,w.
\end{align}\end{proof}

\subsection{Performance of Algorithm \ref{alg:main} for Hierarchical Markets}
 
Here, we analyze the performance of Algorithm \ref{alg:main} for hierarchical matching markets. In the following analysis, we assume that a positive lower bound for $c:=\frac{1}{8}\min_{k\in [n]}\{\Delta,\mu_{m_kw_k}\}$ is known to all the players so that they can choose their mixing parameters in Algorithm \ref{alg:main}. However, as we discuss in Appendix \ref{appx:gap-oracle}, this condition can be relaxed up to some extent by using an uncoordinated and decentralized UCB gap estimation oracle in which the players use their estimated gaps to tune their mixing parameters. 
The main result of this section is presented in Theorem \ref{thm:logarithmic-regret}, which essentially shows that for hierarchical matching markets, Algorithm \ref{alg:main} achieves an expected regret that is logarithmic in $T$ without any coordination. However, before delving into the analysis, we first provide an overview of the main steps in the proof.

\begin{itemize}
    \item First, we define a ``good" event $\Omega$, which represents the set of circumstances where the difference between the accumulated realized rewards, $\eta^t(\hat{L}^t_{m\tilde{w}} - \hat{L}^t_{mw})$, for any $m, w, \tilde{w}$, and $t$, remains close to its conditional expected value. In particular, Lemma \ref{lemm:omega} shows that the event $\Omega$ occurs with very high probability. This allows us to analyze the performance of Algorithm \ref{alg:main} under its conditional mean trajectory while incurring only a small, controllable loss. 
    \item Conditioned on the event $\Omega$, Lemma \ref{lemm:key-hierarchical} shows that, due to the market structure and payoff functions, as time progresses, more pairs in the market are matched according to the underlying hierarchical order with  high probability. Additionally, the size of such nested matchings grows with high probability after a fixed number of steps. 
    \item Finally, in Theorem \ref{thm:logarithmic-regret}, we combine the results of Lemma \ref{lemm:omega} and Lemma \ref{lemm:key-hierarchical} to bound the expected regret of Algorithm \ref{alg:main} by conditioning on the event $\Omega$.  
\end{itemize}

Now we are ready to begin our main analysis. Let us consider the stochastic dynamics of Algorithm \ref{alg:main}, which for any pair $(m,w)$ are described by
\begin{align}\nonumber
&X^{t}_{mw}=\frac{\exp(\eta^t \hat{L}^{t-1}_{mw})}{\sum_{w'}\exp(\eta^t \hat{L}^{t-1}_{mw'})},\cr
&\hat{X}^{t}_{mw}=(1-\gamma^t)X^t_{mw}+\frac{\gamma^t}{n},\cr
&\hat{L}^t_{mw}=\hat{L}^{t-1}_{mw}+\hat{v}^t_{mw}.
\end{align}
Define $Y^t_{mw}=\sum_{\tau=1}^t v_{mw}(\hat{X}^{\tau})$, where $v_{mw}(\hat{X}^{\tau})=\mu_{mw}\prod_{k>_w m}(1-\hat{X}^{\tau}_{mw})$, and consider the event
\begin{align}\nonumber
\Omega:=\Big\{\big|\eta^t\big(\hat{L}^{t-1}_{m\tilde{w}}-\hat{L}^{t-1}_{mw}\big)-\eta^t\big(Y^{t-1}_{m\tilde{w}}-Y^{t-1}_{mw}\big)\big|\leq 2c\sqrt{t}, \ \forall t\ge 1,\ \forall m, w,\tilde{w}\Big\},
\end{align} 
where $c=\frac{1}{8}\min_{k\in [n]}\{\Delta,\mu_{m_kw_k}\}$ is a constant. The following lemma provides a high probability bound for the event $\Omega$ to occur under appropriate choice of stepsize and mixing parameter. 

\begin{lemma}\label{lemm:omega}
Fix an arbitrary $\delta\in (0, 1)$, and suppose each player follows Algorithm \ref{alg:main} with stepsize $\eta^t=\frac{1}{\sqrt{t}}$ and mixing parameter $\gamma^t=M\frac{\log t}{t}$, where $M=\frac{4n}{c}\log \frac{1}{\delta}$. Then, $\mathbb{P}\{\Omega\}\ge 1-\delta.$ 
\end{lemma}
\begin{proof}
Fix an arbitrary man $m$. For any $w, \tilde{w}$, as in \eqref{eq:L-S}, we can write
\begin{align}\label{eq:simpler:L-S}
\eta^t(\hat{L}^{t-1}_{m\tilde{w}}-\hat{L}^{t-1}_{mw})&=\eta^t\sum_{\tau=1}^{t-1}\big(v_{m\tilde{w}}(\hat{X}^{\tau})-v_{mw}(\hat{X}^{\tau})\big)+\eta^t S^{t-1}_{m\tilde{w}}-\eta^t S^{t-1}_{mw}\cr
&=\eta^t(Y^{t-1}_{m\tilde{w}}-Y^{t-1}_{mw})+\eta^t S^{t-1}_{m\tilde{w}}-\eta^t S^{t-1}_{mw},
\end{align}
where $S^{t-1}_{mw}:=\sum_{\tau=1}^{t-1} \big(\hat{v}^{\tau}_{mw}-v_{mw}(\hat{X}^{\tau})\big)\ \forall m,w$. According to Lemma \ref{lemm:unbiased-martingale}, $\{\hat{v}^{\tau}_{mw}-v_{mw}(\hat{X}^{\tau})\}$ is a martingale difference sequence that almost surely satisfies
\begin{align}\nonumber
&|\hat{v}^{\tau}_{mw}-v_{mw}(\hat{X}^{\tau})|\leq \frac{n}{\gamma^{\tau}},\cr 
&\mathbb{E}[\big(\hat{v}^{\tau}_{mw}-v_{mw}(\hat{X}^{\tau})\big)^2|\mathcal{F}^{\tau-1}]\leq \frac{n}{\gamma^{\tau}}.
\end{align}
By taking $u=\frac{n}{\gamma^{t-1}}$, $v^2=\sum_{\tau=1}^{t-1} \frac{n}{\gamma^{\tau}}$, and $x=ct$ in Lemma \ref{lemm:freedman} (see Appendix \ref{appx:freedman}) and noting that the event $\{V^t\leq v^2\}$ holds with probability 1, we obtain 
\begin{align}\label{eq:Bernshtein-log-bound}
\mathbb{P}\{\eta^t|S_{mw}^{t-1}|\ge c\sqrt{t}\}=\mathbb{P}\{|S_{mw}^{t-1}|\ge ct\}\leq 2\exp\Big(\frac{-Mc^2 t^2}{2n(\sum_{\tau=1}^{t-1} \frac{\tau}{\log \tau}+\frac{ct(t-1)}{\log (t-1)})}\Big).
\end{align}
Since $f(\tau):=\frac{\tau}{\log \tau}$ is a concave function for $\tau\ge 100$, using Jensen's inequality, we have
\begin{align}\nonumber
\frac{1}{t-1}\sum_{\tau=1}^{t-1} \frac{\tau}{\log \tau}=\sum_{\tau=1}^{t-1} \frac{1}{t-1}f(\tau)\leq f(\sum_{\tau=1}^{t-1} \frac{1}{t-1}\tau)=\frac{\sum_{\tau=1}^{t-1}\frac{\tau}{t-1}}{\log\big(\sum_{\tau=1}^{t-1}\frac{\tau}{t-1}\big) }=\frac{t}{2\log\frac{t}{2}}. 
\end{align}
Therefore, $\sum_{\tau=1}^{t-1} \frac{\tau}{\log \tau}\leq \frac{t(t-1)}{2\log\frac{t}{2}}\leq \frac{t^2}{\log t}$. Substituting this relation and $\frac{t(t-1)}{\log (t-1)}\leq \frac{t^2}{\log t}$ into \eqref{eq:Bernshtein-log-bound} we get
\begin{align}\nonumber
\mathbb{P}\{\eta^t|S_{mw}^{t-1}|\ge c\sqrt{t}\}\leq 2\exp\Big(\frac{-Mc^2 t^2}{2n(1+c)\frac{t^2}{\log t}}\Big)= 2\exp\Big(\frac{-Mc^2\log t}{2n(1+c)}\Big)\leq \frac{2\delta}{(n\pi t)^2},
\end{align}
where the last inequality is obtained by choosing $M$ sufficiently large (e.g., $M=\frac{4n}{c}\log\frac{1}{\delta}$). Thus, using the union bound, we can write
\begin{align}\label{eq:S-simple-union}
\mathbb{P}\{\eta^t |S_{mw}^{t-1}|< c\sqrt{t}\  \forall t\ge 1, \forall m, w\}\geq 1-\sum_{t=1}^{\infty} \frac{2n^2\delta}{(n\pi t)^2}> 1-\delta.
\end{align}
Finally, using \eqref{eq:simpler:L-S} and \eqref{eq:S-simple-union}, one can see that
\begin{align}\nonumber
\mathbb{P}\{\Omega\}&=\mathbb{P}\Big\{\big|\eta^t\big(\hat{L}^{t-1}_{m\tilde{w}}-\hat{L}^{t-1}_{mw}\big)-\eta^t\big(Y^{t-1}_{m\tilde{w}}-Y^{t-1}_{mw}\big)\big|\leq 2c\sqrt{t}, \ \forall t\ge 1,\ \forall m, w, \tilde{w}\Big\}\cr 
&=\mathbb{P}\Big\{\big|\eta^t S^{t-1}_{m\tilde{w}}-\eta^t S^{t-1}_{mw}\big|\leq c\sqrt{t}, \ \forall t\ge 1,\ \forall m, w, \tilde{w}\Big\}\cr
&\ge\mathbb{P}\{\eta^t |S_{mw}^{t-1}|< c\sqrt{t}\  \forall t\ge 1, \forall m, w\}> 1-\delta. 
\end{align}
\end{proof}

Before we present the main result of this section, we first prove the following key lemma. 

\smallskip
\begin{lemma}\label{lemm:key-hierarchical}
Assume that each man follows Algorithm \ref{alg:main} with stepsize $\eta^t=\frac{1}{\sqrt{t}}$ and mixing parameter $\gamma^t=M\frac{\log t}{t}$, where $M=\frac{4n}{c}\log \frac{1}{\delta}$. Let $a_1\ge \ldots\ge a_n$ be a sequence defined by $a_1=\frac{\Delta}{2}$ and $a_k=\frac{1}{4}\min_{i\in [k]}\big\{\Delta,\mu_{m_iw_i}\big\}, \forall k=2,\ldots,n$. Conditioned on the event $\Omega$, there exists a sequence of time instances $t_1\leq \ldots\leq t_n=O\big(\frac{nM}{c^{n+1}}\log^2(\frac{M}{c})\big)$, such that for any $k\in [n]$, $t\ge t_k$, and $w\neq w_k$, we have 
\begin{align}\label{lemm:statement}
&\hat{X}^t_{m_kw_k}\ge \frac{1-\gamma^t}{1+(n-1)e^{-\eta^t a_k t}}+\frac{\gamma^t}{n},\cr 
&\hat{X}^t_{m_kw}\le \frac{(1-\gamma^t)e^{-\eta^t a_k t}}{1+(n-1)e^{-\eta^t a_k t}}+\frac{\gamma^t}{n}.
\end{align}    
\end{lemma}
\smallskip
\begin{proof}
We use induction on $k\in [n]$ to prove the lemma. For the base of the induction $k=1$, we note that for any $\hat{X}\in [0,1]^{n^2}$, we have
\begin{align}\nonumber
&v_{m_1w_1}(\hat{X})=\mu_{m_1w_1}\prod_{k>_{w_1} m_1}(1-\hat{X}_{kw_1})=\mu_{m_1w_1},\cr
&v_{m_1w}(\hat{X})=\mu_{m_1w}\prod_{k>_{w} m_1}(1-\hat{X}_{kw})\leq \mu_{m_1w}.
\end{align} \vspace{-0.5cm}
Therefore, 
\begin{align}\nonumber
Y^{t-1}_{m_1w_1}-Y^{t-1}_{m_1w}&=\sum_{\tau=1}^{t-1} \big(v_{m_1w_1}(\hat{X}^{\tau})-v_{m_1w}(\hat{X}^{\tau})\big)\ge \sum_{\tau=1}^{t-1} \big(\mu_{m_1w_1}-\mu_{m_1w}\big)\ge \Delta (t-1). 
\end{align}
Define $a_1=\frac{\Delta}{2}$ and let $t_1=4$. Conditioned on the event $\Omega$, for any $t\ge t_1$ and any $w\neq w_1$, we have
\begin{align}\label{eq:base-prelim}
\frac{X^t_{m_1w_1}}{X^t_{m_1w}}&=e^{\eta^t(\hat{L}^{t-1}_{m_1w_1}-\hat{L}^{t-1}_{m_1w})}\ge e^{\eta^t (Y^{t-1}_{m_1w_1}-Y^{t-1}_{m_1w})- 2c\sqrt{t}}\cr
&\ge e^{\eta^t \Delta (t-1)-2c\sqrt{t}}\ge e^{\eta^t a_1 (2t-2)-\frac{1}{2}\eta^t a_1 t}\cr
&= e^{\eta^t a_1 (\frac{3}{2}t-2)} \ge e^{\eta^t a_1t},
\end{align}
where the third inequality holds because $c\leq \frac{a_1}{4}$, and the last inequality holds by $t\ge t_1$. Moreover, 
\begin{align} \label{eq:base-prelim-two}
1=X^t_{m_1w_1}+\sum_{w\neq w_1}X^t_{m_1w}\leq \big(1+(n-1)e^{-\eta^t a_1 t}\big)X^t_{m_1w_1}.
\end{align}
Combining \eqref{eq:base-prelim} and \eqref{eq:base-prelim-two} and using $\hat{X}^t_{mw}=(1-\gamma^t)X^t_{mw}+\frac{\gamma^t}{n}$, for any $t\ge t_1$ and $w\neq w_1$, we have
\begin{align}\label{eq:induction-base}
&\hat{X}^t_{m_1w_1}\ge \frac{1-\gamma^t}{1+(n-1)e^{-\eta^t a_1 t}}+\frac{\gamma^t}{n},\cr
&\hat{X}^t_{m_1w}\le \frac{(1-\gamma^t)e^{-\eta^t a_1 t}}{1+(n-1)e^{-\eta^t a_1 t}}+\frac{\gamma^t}{n}. 
\end{align} 
Next, by induction hypothesis, suppose that the statement holds for the first $k-1$ pairs $(m_{\ell},w_{\ell}), \ell=1,\ldots,k-1$ with corresponding sequences $a_1\ge a_2\ldots\ge a_{k-1}$ and $t_1\leq t_2\leq\ldots\leq t_{k-1}$, and let us consider the $k$th pair $(m_k,w_k)$ of the market. Notice that for any $\hat{X}\in [0,1]^{n^2},$ we have 
\begin{align}\label{eq:general-two-cases}
v_{m_{k}w_{k}}(\hat{X})=\mu_{m_kw_k}\prod_{m_{\ell}>_{w_k} m_k}(1-\hat{X}_{m_{\ell}w_k})\ge \mu_{m_kw_k}\prod_{\ell=1}^{k-1} (1-\hat{X}_{m_{\ell}w_k}). 
\end{align}
We consider two cases: 

\noindent
{\bf Case I:} If $w=w_{i}$ for some $i\in \{1,\ldots,k-1\}$, then
\begin{align}\label{eq:case-i}
v_{m_kw}(\hat{X})=v_{m_kw_i}(\hat{X})=\mu_{m_kw_i}\prod_{m_{\ell}>_{w_i} m_k}(1-\hat{X}_{m_{\ell}w_i})\leq \mu_{m_kw_i}(1-\hat{X}_{m_{i}w_{i}}). 
\end{align}
Therefore, using \eqref{eq:general-two-cases} and \eqref{eq:case-i} we can write
\begin{align}\nonumber
&Y^{t-1}_{m_kw_k}-Y^{t-1}_{m_kw}=\sum_{\tau=t_{k-1}}^{t-1} \big(v_{m_kw_k}(\hat{X}^{\tau})-v_{m_kw_i}(\hat{X}^{\tau})\big)+\sum_{\tau=1}^{t_{k-1}-1} \big(v_{m_kw_k}(\hat{X}^{\tau})-v_{m_kw_i}(\hat{X}^{\tau})\big)\cr
&\ge \sum_{\tau=t_{k-1}}^{t-1}\Big(\mu_{m_kw_k}\prod_{\ell=1}^{k-1}(1-\hat{X}^{\tau}_{m_{\ell}w_k})-\mu_{m_kw_i}(1-\hat{X}^{\tau}_{m_iw_i})\Big)+\sum_{\tau=1}^{t_{k-1}-1} \big(v_{m_kw_k}(\hat{X}^{\tau})-v_{m_kw_i}(\hat{X}^{\tau})\big)\cr
&\ge \sum_{\tau=t_{k-1}}^{t-1}\Big(\mu_{m_kw_k}\big(1-\sum_{\ell=1}^{k-1}(\frac{\gamma^{\tau}}{n}+\frac{(1-\gamma^{\tau})e^{-\eta^{\tau} a_{\ell} \tau}}{1+(n-1)e^{-\eta^{\tau} a_{\ell} \tau}})\big)-\mu_{m_kw_i}\big(1-\frac{\gamma^{\tau}}{n}-\frac{1-\gamma^{\tau}}{1+(n-1)e^{-\eta^{\tau} a_i \tau}}\big)\Big)-t_{k-1} \cr 
&\ge \sum_{\tau=t_{k-1}}^{t-1}\Big(\mu_{m_kw_k}-\gamma^{\tau}-\sum_{\ell=1}^{k-1}\frac{(1-\gamma^{\tau})e^{-\eta^{\tau} a_{\ell} \tau}}{1+(n-1)e^{-\eta^{\tau} a_{\ell} \tau}}-\frac{(n-1)e^{-\eta^{\tau} a_{i} \tau}+\gamma^{\tau}}{1+(n-1)e^{-\eta^{\tau} a_i \tau}}\Big)-t_{k-1}\cr
&\ge \sum_{\tau=t_{k-1}}^{t-1}\Big(\mu_{m_kw_k}-\gamma^{\tau}-\sum_{\ell=1}^{k-1}e^{-\eta^{\tau} a_{\ell} \tau}-(n-1)e^{-\eta^{\tau} a_{i} \tau}-\gamma^{\tau}\Big)-t_{k-1}\cr
&\ge \sum_{\tau=t_{k-1}}^{t-1} \big(\mu_{m_kw_k}-2\gamma^{\tau}-2n e^{-\eta^{\tau} a_{k-1} \tau}\big)-t_{k-1},
\end{align}where the second inequality holds by the induction hypothesis for $\tau\ge t_{k-1}$, and the last inequality holds because $a_1\ge \ldots\ge a_{k-1}$. Using the specific choice of stepsize and mixing parameters $\eta^t=\frac{1}{\sqrt{t}}$ and $\gamma^t=M\frac{\log t}{t}$, for any $t\ge t_{k-1}$, we have 
\begin{align}\nonumber
Y^{t-1}_{m_kw_k}-Y^{t-1}_{m_kw}&\ge \sum_{\tau=t_{k-1}}^{t-1} \big(\mu_{m_kw_k}-2\gamma^{\tau}-2ne^{-\eta^{\tau} a_{k-1} \tau}\big)-t_{k-1}\cr
&\ge \mu_{m_kw_k}(t-t_{k-1})-2M\int_{1}^{t}\frac{\log \tau}{\tau}d\tau-2n\int_{0}^{\infty}e^{-a_{k-1} \sqrt{\tau}}d\tau-t_{k-1}\cr
&\ge \mu_{m_kw_k}t-2M \log^2 t-\frac{4n}{a^2_{k-1}}-2t_{k-1}. 
\end{align}
In particular, for any $t\ge \frac{4}{\mu_{m_kw_k}}(2M\log^2 \frac{4M}{\mu_{m_kw_k}}+\frac{2n}{a^2_{k-1}}+t_{k-1})$, we have
\begin{align}\label{eq:case-i-conclusion}
Y^{t-1}_{m_kw_k}-Y^{t-1}_{m_kw}\ge \frac{\mu_{m_kw_k}}{2}t.   
\end{align}

\noindent 
{\bf Case II:} If $w\notin \{w_{1},\ldots,w_{k}\},$ then $\mu_{m_kw_k}\ge \mu_{m_kw}+\Delta$. Using \eqref{eq:general-two-cases} and the fact that $\mu_{m_kw}\ge v_{m_kw}(\hat{X})$ for any $\hat{X}\in [0,1]^{n^2}$, we can write
\begin{align}\nonumber
Y^{t-1}_{m_kw_k}-Y^{t-1}_{m_kw}&=\sum_{\tau=t_{k-1}}^{t-1} \big(v_{m_kw_k}(\hat{X}^{\tau})-v_{m_kw}(\hat{X}^{\tau})\big)+\sum_{\tau=1}^{t_{k-1}-1} \big(v_{m_kw_k}(\hat{X}^{\tau})-v_{m_kw}(\hat{X}^{\tau})\big)\cr
&\ge \sum_{\tau=t_{k-1}}^{t-1} \Big(\mu_{m_kw_k}\prod_{\ell=1}^{k-1}(1-\hat{X}^{\tau}_{m_{\ell}w_k})-\mu_{m_kw}\Big)-t_{k-1}\cr
&\ge \sum_{\tau=t_{k-1}}^{t-1} \Big(\mu_{m_kw_k}\big(1-\sum_{\ell=1}^{k-1}\hat{X}^{\tau}_{m_{\ell}w_k}\big)-\mu_{m_kw}\Big)-t_{k-1}\cr
&\ge \sum_{\tau=t_{k-1}}^{t-1}\Big((\mu_{m_kw_k}-\mu_{m_kw})-\mu_{m_kw_k}\sum_{\ell=1}^{k-1}\big(\frac{\gamma^{\tau}}{n}+\frac{(1-\gamma^t)e^{-\eta^{\tau} a_{\ell} \tau}}{1+(n-1)e^{-\eta^{\tau} a_{\ell} \tau}}\big)\Big)-t_{k-1}\cr 
&\ge \sum_{\tau=t_{k-1}}^{t-1}\Big(\Delta-\gamma^{\tau}-\sum_{\ell=1}^{k-1}e^{-\eta^{\tau} a_{\ell} \tau}\Big)-t_{k-1}\cr
&\ge \sum_{\tau=t_{k-1}}^{t-1}\Big(\Delta-\gamma^{\tau}-ne^{-\eta^{\tau} a_{k-1} \tau}\Big)-t_{k-1}.
\end{align}
Thus, for any integer $t\ge t_{k-1}$ and the choice of parameters $\eta^t=\frac{1}{\sqrt{t}}$ and $\gamma^t=M\frac{\log t}{t}$, we have 
\begin{align}\nonumber
Y^{t-1}_{m_kw_k}-Y^{t-1}_{m_kw}&\ge  \sum_{\tau=t_{k-1}}^{t-1}\Big(\Delta-\gamma^{\tau}-ne^{-\eta^{\tau} a_{k-1} \tau}\Big)-t_{k-1}\cr
&\ge \Delta(t-t_{k-1})-M\log^2 t-\frac{4n}{a^2_{k-1}}-t_{k-1}\cr
&\ge \Delta t-M\log^2 t-\frac{4n}{a^2_{k-1}}-2t_{k-1}.
\end{align}
Therefore, for any $t\ge \frac{4}{\Delta}(M\log^2 (\frac{2M}{\Delta})+\frac{2n}{a^2_{k-1}}+t_{k-1})$, we have 
\begin{align}\label{eq:case-II-conclusion}
Y^{t-1}_{m_kw_k}-Y^{t-1}_{m_kw}\ge\frac{\Delta}{2}t.
\end{align}
As a result, by taking
\begin{align}\label{eq:t_k}
&t_k:=\frac{4}{\min\{\Delta,\mu_{m_kw_k}\}}\Big(2M\log^2(\frac{2M}{\min\{\Delta,\mu_{m_kw_k}\}})+\frac{2n}{a^2_{k-1}}+t_{k-1}\Big),\cr
&a_{k}:=\frac{1}{4}\min_{i\in [k]}\Big\{\Delta, \mu_{m_iw_i}\Big\}, 
\end{align}
we have $a_{k-1}\ge a_{k}, t_k\ge t_{k-1}$. In particular, for $t\ge t_k$ both relations \eqref{eq:case-i-conclusion} and \eqref{eq:case-II-conclusion} hold and we have 
\begin{align}\label{eq:final-Y-a}
Y^{t-1}_{m_kw_k}-Y^{t-1}_{m_kw}\geq \min \{\frac{\Delta}{2}, \frac{\mu_{m_kw_k}}{2}\} t\ge 2a_k t,\ \forall t\ge t_k, \forall w\neq w_k.
\end{align} 
Now, using a similar argument as in the base case, for any $t\ge t_k$ and $w\neq w_k$, we have
\begin{align}\nonumber
\frac{X^t_{m_kw_k}}{X^t_{m_kw}}&=e^{\eta^t(\hat{L}^{t-1}_{m_kw_k}-\hat{L}^{t-1}_{m_kw})}\ge e^{\eta^t (Y^{t-1}_{m_kw_k}-Y^{t-1}_{m_kw})- 2c\sqrt{t}}\cr
&\ge e^{\eta^t 2a_k t-2c\sqrt{t}}\ge e^{\eta^t 2a_k t-\eta^t a_k t}= e^{\eta^t a_kt},
\end{align}
where the second inequality uses \eqref{eq:final-Y-a}, and the last inequality holds because $c\!=\!\min_{k\in [n]}\{\frac{\Delta}{8},\frac{\mu_{m_kw_k}}{8}\}\leq \frac{a_k}{2}$. The above relation, together with similar arguments as in deriving \eqref{eq:base-prelim-two} and \eqref{eq:induction-base} in which $a_1$ is replaced by $a_k$, gives us the desired result \eqref{lemm:statement}. 

Finally, using relation \eqref{eq:t_k} recursively, and noting that for any $k$ we have $\frac{4}{\min\{\Delta,\mu_{m_kw_k}\}}\leq \frac{1}{c}$ and $\frac{2}{a_k}\leq \frac{1}{c}$, one can obtain an explicit upper bound for $t_n$ as
\begin{align}\nonumber
t_n&\leq \frac{1}{c}\Big(2M\log^2(\frac{M}{c})+\frac{n}{c^2}+t_{n-1}\Big)\leq \cdots\cr
&\leq \Big(\sum_{k=1}^{n-1} \frac{1}{c^k}\Big)\cdot \Big(2M\log^2(\frac{M}{c})+\frac{n}{c^2}+t_1\Big)\cr
&=O\Big(\frac{n M}{c^{n+1}}\log^2(\frac{M}{c})\Big).    
\end{align}
\end{proof}

\begin{theorem}\label{thm:logarithmic-regret}
Assume that each man follows Algorithm \ref{alg:main} with $\eta^t=\frac{1}{\sqrt{t}}$ and $\gamma^t=M\frac{\log t}{t}$, where $M=\frac{4n}{c}\log T$ and $c=\frac{1}{8}\min_{k\in [n]}\{\Delta,\mu_{m_kw_k}\}$. Then, the expected regret of Algorithm \ref{alg:main} in hierarchical matching markets is at most $R(T)=\tilde{O}\Big(\frac{n^3}{c^{n+2}}\log T+\frac{n^2}{c}\log^3 T\Big)$, where $\tilde{O}$ hides the $\log(\log T)$ terms.  
\end{theorem}
\begin{proof}
Using the definition of the expected regret, we can write 
\begin{align}\label{eq:expected-regret-first}
\mathbb{E}[R(T)]&= \sum_{t=1}^T \mathbb{E}[\boldsymbol{1}_{\{\exists k: \alpha^{t}_{m_k}\neq w_k\}}]\cr
&\leq \sum_{t=1}^T\sum_{k=1}^n \mathbb{E}[\boldsymbol{1}_{\{\alpha^{t}_{m_k}\neq w_k\}}]\cr 
&= \sum_{t=1}^T\sum_{k=1}^n \mathbb{E}\big[\mathbb{E}[\boldsymbol{1}_{\{\alpha^{t}_{m_k}\neq w_k\}}|\hat{X}^{t}_{m_k}]\big]\cr
&=\sum_{t=1}^T\sum_{k=1}^n \mathbb{E}\big[\mathbb{P}\{\alpha^{t}_{m_k}\neq w_k|\hat{X}^{t}_{m_k}\}\big]\cr
&=\sum_{t=1}^T\sum_{k=1}^n \mathbb{E}\big[1-\hat{X}^{t}_{m_kw_k}\big]\cr
&\leq nt_n+\sum_{t=t_n+1}^{T}\sum_{k=1}^n \mathbb{E}\big[1-\hat{X}^{t}_{m_kw_k}\big]
\end{align}
On the other hand, by conditioning on the event $\Omega$, and using the fact that $\mathbb{P}\{\bar{\Omega}\}\leq \delta$, we have
\begin{align}\label{eq:expect-condition-omega}
\mathbb{E}\big[1-\hat{X}^{t}_{m_kw_k}\big]&=\mathbb{P}\{\Omega\}\mathbb{E}\big[1-\hat{X}^{t}_{m_kw_k}|\Omega\big]+\mathbb{P}\{\bar{\Omega}\}\mathbb{E}\big[1-\hat{X}^{t}_{m_kw_k}|\bar{\Omega}\big]\cr
&\leq \mathbb{E}\big[1-\hat{X}^{t}_{m_kw_k}|\Omega\big]+\delta 
\end{align}
Substituting \eqref{eq:expect-condition-omega} into \eqref{eq:expected-regret-first} and using Lemma \ref{lemm:key-hierarchical}, we obtain
\begin{align}
\mathbb{E}[R(T)]&\leq nt_n+\sum_{t=t_n+1}^{T}\sum_{k=1}^n \mathbb{E}\big[1-\hat{X}^{t}_{m_kw_k}|\Omega\big]+\delta nT\cr
&\leq O\Big(\frac{n^2M}{c^{n+1}}\log^2(\frac{M}{c})\Big)+\sum_{t=t_n+1}^T\sum_{k=1}^n \big(1-\frac{\gamma^t}{n}-\frac{1-\gamma^t}{1+(n-1)e^{-\eta^t a_k t}}\big)+\delta nT\cr
&\leq O\Big(\frac{n^2M}{c^{n+1}}\log^2(\frac{M}{c})\Big)+\sum_{t=t_n+1}^T\sum_{k=1}^n\frac{(n-1)e^{-\eta^t a_k t}+\gamma_t}{1+(n-1)e^{-\eta^t a_k t}}+\delta nT\cr
&\leq O\Big(\frac{n^2M}{c^{n+1}}\log^2(\frac{M}{c})\Big)+n^2\sum_{t=1}^T e^{-\eta^t a_{n} t}+n\sum_{t=1}^T\gamma_t+\delta nT\cr
&\leq O\Big(\frac{n^2M}{c^{n+1}}\log^2(\frac{M}{c})\Big)+\frac{2n^2}{c^2}+nM\log^2 T+\delta nT\cr
&=O\Big(\frac{n^2M}{c^{n+1}}\log^2(\frac{M}{c})+nM\log^2 T+\delta nT\Big), 
\end{align}
where the fourth inequality uses the fact that $a_{n}:=\min_{k\in [n]}\{a_k\}$, and the last inequality is obtained by choosing $\gamma^t=M\frac{\log t}{t}, \eta^t=\frac{1}{\sqrt{t}}$, and noting that $c\leq a_n$. Finally, since $\delta\in (0,1)$ can be chosen arbitrarily, by choosing $\delta=\frac{1}{T}$ and noting that $M=\frac{4n}{c}\log \frac{1}{\delta}=\frac{4n}{c}\log T$, we can write
\begin{align}\nonumber
\mathbb{E}[R(T)]&\leq O\Big(\frac{n^2M}{c^{n+1}}\log^2(\frac{M}{c})+nM\log^2 T+\delta nT\Big) \cr
&=O\Big(\frac{n^3}{c^{n+2}}\log T\cdot\log^2(\log T)+\frac{n^2}{c}\log^3 T\Big).
\end{align}
\end{proof}

\section{Exponential Local Convergence for General Matching Markets} \label{sec:dual-mirror}

In this section, we consider learning a stable matching in general matching markets (i.e., without any hierarchical assumption) and show that the same decentralized and uncoordinated algorithm EXP (Algorithm \ref{alg:main}) converges locally at an exponential rate to a stable matching when players' strategies get sufficiently close to one of the stable matchings. To that end, we first consider the following technical lemma, which shows that if a strategy profile is sufficiently close to a stable matching, then the reward of choosing an action according to that stable matching is strictly the best decision. 

\smallskip   
\begin{lemma}\label{lemm:local_neighbor}
Let $X^*$ be a pure NE of the stable matching game and set $c=\frac{1}{8}\min\{\Delta,\mu_{\min}\}$. For any strategy profile $x$ that satisfies $\|x-X^*\|_1\leq \frac{c}{\mu_{\max}}$, we have $v_{mw^*}(x)-v_{mw}(x)>c\ \forall w\neq w^*$, where $w^*$ is the woman that $m$ is matched to her under the pure NE $X^*$. 
\end{lemma}
\begin{proof}
Please see Appendix \ref{sec:local_neighbor} for the proof.\end{proof}

Next, we consider the following lemma that will be used to capture the local behavior of the dynamics of Algorithm \ref{alg:main} around a pure NE. The proof of this lemma can be found in Appendix \ref{sec:proof-lemma}.  

\smallskip
\begin{lemma}\label{lemm:tech-bound}
Suppose $X^*$ is a pure NE and $\{X^t\}$ be the sequence of iterates generated by Algorithm \ref{alg:main}. Moreover, let $A^{t-1}= \sum_{\tau=1}^{t-1}\big(v_{mw^*}(X^{\tau})-v_{mw}(X^{\tau})\big)$, where $w^*$ is the woman that man $m$ is matched to her under the pure NE $X^*$. Then, the following statements hold:
\begin{itemize}
\item i) If $\|X^t-X^*\|_1\leq \frac{c}{250n^2}$, then $\eta^{t}\hat{L}^{t-1}_{mw^*}-\eta^{t}\hat{L}^{t-1}_{mw}\ge \ln (\frac{n^2}{2c})+6 \ \forall m,w\neq w^*$.
\item ii) If $\eta^{t}\hat{L}^{t-1}_{mw^*}-\eta^{t}\hat{L}^{t-1}_{mw}\ge \ln(\frac{n^2}{2c})\ \forall m,w\neq w^*$, then $\|X^{t}-X^*\|_1\leq c$.
\item iii) Let $t_0$ be the first time such that $1/\eta^{t_0+1}-1/\eta^{t_0}\leq c/(\ln(\frac{n^2}{2c})+3)$. If for some $t\ge t_0$ we have $\eta^{t} A^{t-1}\ge \ln(\frac{n^2}{2c})+3$ and $v_{mw^*}(X^{t})-v_{mw}(X^{t})\ge c\ \forall m,w\neq w^*$, then $\eta^{t+1} A^{t}\ge \ln(\frac{n^2}{2c})+3$.
\end{itemize} 
\end{lemma}

Now we are ready to state the main result of this section. 

\smallskip
\begin{theorem}\label{thm:local-exponential}
Fix an arbitrary $\delta\in (0, 1)$, and suppose each player follows Algorithm \ref{alg:main} with some mixing sequence $\{\gamma^t\}$ and decreasing stepsize sequence $\{\eta^t\}$ that satisfie the following conditions:
\begin{align}\label{eq:loca-step-condition}
\Big(\sum_{\tau=1}^{t-1} \frac{\gamma^{t-1}}{\gamma^{\tau}}\Big)^{-1}\leq \eta^{t}\leq \min \left\{\Big(2n^2\sum_{\tau=1}^{t-1} \gamma^{\tau}\Big)^{-1},\ \Big(8n\ln(\frac{n\pi}{\sqrt{\delta}}t)\sum_{\tau=1}^{t-1} \frac{1}{\gamma^{\tau}}\Big)^{-1/2}\right\}. 
\end{align}
Let $t_0$ be the first time such that $1/\eta^{t_0+1}-1/\eta^{t_0}\leq c/(\ln(\frac{n^2}{2c})+3)$ and $\|X^{t_0}-X^*\|\leq \frac{c}{250n^2}$. Then with probability at least $1-\delta$, the sequence $\{X^{t}\}$ converges to the pure NE $X^*$ exponentially fast, i.e.,
\begin{align}\nonumber
\mathbb{P}\Big\{\|X^{t+1}-X^*\|_1\leq 41 n\exp(-ct\eta^{t+1})\ \forall t\ge t_0\ \big|\ \|X^{t_0}-X^*\|_1\leq \frac{c}{250n^2} \Big\}\ge 1-\delta.
\end{align}  
\end{theorem}
\begin{proof}
Fix an arbitrary man $m$, and assume $X^*_{mw^*}=1$. For any $w\neq w^*$, we have
\begin{align}\label{eq:prem:L-S}
\eta^t&\hat{L}^{t-1}_{mw^*}-\eta^t\hat{L}^{t-1}_{mw}=\eta^t\sum_{\tau=1}^{t-1}\big(\hat{v}^{\tau}_{mw^*}-\hat{v}^{\tau}_{mw}\big)\cr
&=\eta^t\sum_{\tau=1}^{t-1}\big(v_{mw^*}(\hat{X}^{\tau})-v_{mw}(\hat{X}^{\tau})\big)+\eta^t\sum_{\tau=1}^{t-1}\big(\hat{v}^{\tau}_{mw^*}-v_{mw^*}(\hat{X}^{\tau})\big)-\eta^t\sum_{\tau=1}^{t-1}\big(\hat{v}^{\tau}_{mw}-v_{mw}(\hat{X}^{\tau})\big)\cr
&=\eta^t\sum_{\tau=1}^{t-1}\big(v_{mw^*}(\hat{X}^{\tau})-v_{mw}(\hat{X}^{\tau})\big)+\eta^t S^{t-1}_{mw^*}-\eta^t S^{t-1}_{mw},
\end{align} 
where $S^{t-1}_{mw}:=\sum_{\tau=1}^{t-1} \big(\hat{v}^{\tau}_{mw}-v_{mw}(\hat{X}^{\tau})\big)$. Moreover, for any $m$ and $w$, using the mean-value theorem, there exists a $\theta=\lambda \hat{X}^{\tau}+(1-\lambda)X^{\tau}$ such that 
\begin{align}\nonumber
|v_{mw}(\hat{X}^{\tau})-v_{mw}(X^{\tau})|&=|\nabla v_{mw}(\theta)\cdot (\hat{X}^{\tau}-X^{\tau})|\leq \|\nabla v_{mw}(\theta)\|_{\infty}\|\hat{X}^{\tau}-X^{\tau}\|_1\cr
&\leq n \|\hat{X}^{\tau}-X^{\tau}\|_1=n\gamma^{\tau}\sum_m\|X_m^{\tau}-\frac{1}{n}\boldsymbol{1}\|_1 \leq n^2\gamma^{\tau}.
\end{align}
Substituting this relation into \eqref{eq:prem:L-S} and using the triangle inequality, we can write
\begin{align}\label{eq:L-S}
\big|\eta^t\hat{L}^{t-1}_{mw^*}-\eta^t\hat{L}^{t-1}_{mw}-\eta^t\sum_{\tau=1}^{t-1}\big(v_{mw^*}(X^{\tau})-v_{mw}(X^{\tau})\big)\big|&\leq \eta^t |S^{t-1}_{mw^*}|+\eta^t |S^{t-1}_{mw}|+2n^2\eta^t\sum_{\tau=1}^{t-1} \gamma^{\tau}\cr
&\leq \eta^t |S^{t-1}_{mw^*}|+\eta^t |S^{t-1}_{mw}|+1,
\end{align}
where the last inequality follows from the upper bound on the choice of stepsize $\eta^t$. 

Next, notice that according to Lemma \ref{lemm:unbiased-martingale}, $\{\hat{v}^{\tau}_{mw}-v_{mw}(\hat{X}^{\tau})\}$ is a martingale difference sequence that almost surely satisfies 
\begin{align}\nonumber
|\hat{v}^{\tau}_{mw}-v_{mw}(\hat{X}^{\tau})|\leq \frac{n}{\gamma^{\tau}},\ \ \ \ \mathbb{E}[\big(\hat{v}^{\tau}_{mw}-v_{mw}(\hat{X}^{\tau})\big)^2|\mathcal{F}^{\tau-1}]\leq \mathbb{E}[(\hat{v}_{mw}^{\tau})^2|\mathcal{F}^{\tau-1}]\leq \frac{n}{\gamma^{\tau}}.
\end{align}
Therefore, by taking $u=\frac{n}{\gamma^{t-1}}$, $v^2=\sum_{\tau=1}^{t-1} \frac{n}{\gamma^{\tau}}$, and $x=\frac{1}{\eta^t}$ in Lemma \ref{lemm:freedman}, we obtain
\begin{align}\nonumber
\mathbb{P}\{\eta^t |S_{mw}^{t-1}|\ge 1\}\leq 2\exp\Big(-\frac{(1/\eta^t)^2}{2n(\sum_{\tau=1}^{t-1} \frac{1}{\gamma^{\tau}}+\frac{1}{\eta^t\gamma^{t-1}})}\Big)\leq 2\exp\big(\frac{-(1/\eta^t)^2)}{4n\sum_{\tau=1}^{t-1} \frac{1}{\gamma^{\tau}}}\big),
\end{align} 
where the second inequality holds because using the lower bound on $\eta^t$ we have $\frac{1}{\eta^{t}\gamma^{t-1}}\leq \sum_{\tau=1}^{t-1} \frac{1}{\gamma^{\tau}}$. Since by the choice of stepsize $\frac{(1/\eta^t)^2}{\sum_{\tau=1}^{t-1} 1/\gamma^{\tau}}\ge 8n\ln(\frac{n\pi}{\sqrt{\delta}}t)$, we get $\mathbb{P}\{\eta^t |S_{mw}^{t-1}|\ge 1\}\leq \frac{2\delta}{(n\pi t)^2}$. In particular, using the union bound, the probability that $\eta^t |S_{mw}^{t-1}|< 1$ for any $t\ge 1$ and any $m, w$, is at least
\begin{align}\label{eq:S-m-w-union}
\mathbb{P}\{\eta^t |S_{mw}^{t-1}|< 1\  \forall t\ge 1, \forall m, w\}\geq 1-\sum_{t=1}^{\infty} \frac{2n^2\delta}{(n\pi t)^2}> 1-\delta.
\end{align}
Let $A^{t-1}= \sum_{\tau=1}^{t-1}\big(v_{mw^*}(X^{\tau})-v_{mw}(X^{\tau})\big)$. Using \eqref{eq:L-S} and \eqref{eq:S-m-w-union}, with probability at least $1-\delta$, we have
\begin{align}\label{eq:L-A}
|\eta^t\hat{L}^{t-1}_{mw^*}-\eta^t\hat{L}^{t-1}_{mw}-\eta^{t}A^{t-1}|\leq 3, \ \forall t\ge 1,\ \forall m, w.
\end{align}    

Next, let us assume $\|X^t-X^*\|_1\in \frac{c}{250n^2}$. Then, for any $m, w\neq w^*$, we have $v_{mw^*}(X^{t})-v_{mw}(X^{t})\ge c$ and $\eta^{t}\hat{L}^{t-1}_{mw^*}-\eta^{t}\hat{L}^{t-1}_{mw}\ge \ln(\frac{n^2}{2c})+6$, by Lemma \ref{lemm:local_neighbor} and Lemma \ref{lemm:tech-bound} (part i), respectively. Using relation \eqref{eq:L-A}, it means that $\eta^{t}A^{t-1} \ge \ln(\frac{n^2}{2c})+3$, and thus by Lemma \ref{lemm:tech-bound} (part iii), we have $\eta^{t+1} A^{t}\ge \ln(\frac{n^2}{2c})+3$. This in view of Lemma \ref{lemm:tech-bound} (part ii) shows that $\|X^{t+1}-X^*\|_1\leq c$. Similarly, Lemma \ref{lemm:local_neighbor} implies $v_{mw^*}(X^{t+1})-v_{mw}(X^{t+1})\ge c\ \forall m,w\neq w^*$, which in view of Lemma \ref{lemm:tech-bound} (part iii) shows that $\eta^{t+2} A^{t+1}\ge \ln(\frac{n^2}{2c})+3$. Hence, using relation \eqref{eq:L-A}, $\eta^{t+2}\hat{L}^{t+1}_{mw^*}-\eta^{t+2}\hat{L}^{t+1}_{mw}\ge \ln(\frac{n^2}{2c})\ \forall m,w\neq w^*$, and thus $\|X^{t+2}-X^*\|_1\leq c$ by Lemma \ref{lemm:tech-bound} (part ii). 

The above inductive argument shows that if $\|X^{t_0}-X^*\|\leq \frac{2c}{n^2}$ for some $t_0\ge 1$ such that $1/\eta^{t_0+1}-1/\eta^{t_0}\leq c/(\ln(\frac{n^2}{2c})+3)$, then with probability at least $1-\delta$, we have $\|X^{t}-X^*\|\leq c\ \forall t\ge t_0$, and in particular, using Lemma \ref{lemm:local_neighbor}, $v_{mw^*}(X^{t})-v_{mw}(X^{t})\ge c\ \forall m,w, t\ge t_0$ . Therefore, with probability at least $1-\delta$, for any $m, w\neq w^*, t\ge t_0$, we have 
\begin{align}\nonumber
\frac{X^{t+1}_{mw^*}}{X^{t+1}_{mw}}=\exp(\eta^{t+1}\hat{L}^t_{mw^*}-\eta^{t+1}\hat{L}^t_{mw})\ge \exp(\eta^{t+1}A^t-3)\ge \exp(ct\eta^{t+1}-3).  
\end{align}
Therefore, $\|X^{t+1}_{m}-X^*_{m}\|_1=2(1-X^{t+1}_{mw^*})\leq 2n\exp(-ct\eta^{t+1}+3)$, which shows that
\begin{align}\nonumber
\mathbb{P}\Big\{\|X^{t+1}-X^*\|_1\leq 2n\exp(-ct\eta^{t+1}+3)\ \forall t\ge t_0\big|\ \|X^{t_0}-X^*\|_1\leq \frac{c}{250n^2}\Big\}\ge 1-\delta.
\end{align}
\end{proof}

\begin{remark}
    In the statement of Theorem \ref{thm:local-exponential}, the condition that $t\ge t_0$  is not restrictive because if Algorithm \ref{alg:main} converges to a pure NE, then eventually for some $t_0$, we have $1/\eta^{t_0+1}-1/\eta^{t_0}\leq c/(\ln (\frac{n^2}{2c})+3)$ and  $\|X^{t_0}-X^*\|\leq \frac{c}{250n^2}$. In particular, $t_0$ is a constant that scales polynomially in terms of other parameters (see, e.g., Remark \ref{re:specific-stepsize}). We note that it is possible to further relax this condition by a constant scaling of $\eta^t$, i.e., by setting $\eta^t\leftarrow \frac{1}{c}(\ln (\frac{n^2}{2c})+3)\eta^t$ and rescaling all other parameters such as the mixing sequence $\gamma^t$ accordingly. However, that requires the players to know $c=\frac{1}{8}\min\{\Delta,\mu_{\min}\}$ to choose their stepsize $\eta^t$, while knowing $c$ is not required in the stepsize condition given in \eqref{eq:loca-step-condition}.    
\end{remark}

\begin{remark}\label{re:specific-stepsize}
In fact, there exist feasible choices for $\eta^{t}$ and $\gamma^t$ that satisfy the stepsize condition \eqref{eq:loca-step-condition} in Theorem \ref{thm:local-exponential}. For instance, if for some $\alpha,\beta\in (0,1), \ \alpha> \max\{1-\beta, \frac{1+\beta}{2}\}$, we choose $\eta^{t}=\Theta\big(\frac{1}{\ln( t/\sqrt{\delta})t^{\alpha}}\big)$ and $\gamma^{t}=\Theta(\frac{1}{t^{\beta}})$, then condition \eqref{eq:loca-step-condition} is satisfied for any $t_0\ge (\frac{\alpha b}{c})^{1/(1-\alpha)}$, where $b=\ln(\frac{n^2}{2c})+3$. More specifically, by choosing $\gamma^{t}=\frac{1}{t^{1/3}}$ and $\eta^{t}=\frac{1}{t^{3/4}}$, the exponential convergence rate of Theorem \ref{thm:local-exponential} would be in the order of $\exp(-ct^{1/4})$ for any $t_0\ge (\frac{b}{c})^4$. We note that condition \eqref{eq:loca-step-condition} provides only a sufficient condition to achieve local exponential convergence and we did not attempt to fully characterize all feasible choices of stepsize/mixing parameters. However, there could be other choices of stepsize that do not satisfy condition \eqref{eq:loca-step-condition} while still achieving local exponential convergence.
\end{remark}

Although Theorem \ref{thm:local-exponential} shows that Algorithm \ref{alg:main} achieves an exponential convergence rate, its convergence guarantee holds in a local sense, i.e., when the dynamics get sufficiently close to a neighborhood of a stable matching.
In the following theorem, we complement this result by showing that if the dynamics of Algorithm \ref{alg:main} converge with positive probability, it must converge to a mixed NE of the stable matching game. 

\smallskip
\begin{theorem}\label{thm:positive-prob-mixed-NE}
Under the same stepsize assumptions of Theorem \ref{thm:local-exponential}, if Algorithm \ref{alg:main} converges to a limit point $X^*$ with positive probability, then $X^*$ must be a NE of the stable matching game. 
\end{theorem}
\begin{proof}
Using identical arguments as in the proof of Theorem \ref{thm:local-exponential}, with probability at least $1-\delta$,
\begin{align}\label{eq:pp-L-A}
\big|\eta^t\hat{L}^{t-1}_{mw'}-\eta^t\hat{L}^{t-1}_{mw}-\eta^{t}A^{t-1}\big|\leq 3, \ \forall t\ge 1,\ \forall m, w\ne w',
\end{align}
where we recall that $A^{t-1}= \sum_{\tau=1}^{t-1}\big(v_{mw'}(X^{\tau})-v_{mw}(X^{\tau})\big)$. Now let $\Omega$ be the event such that $X^{\tau}\to X^*$, and assume $X^*_{mw}>0$ for some $m,w$. Conditioned on the event $\Omega$, we have
\begin{align}\nonumber
\eta^{t}A^{t-1}\leq \eta^t\hat{L}^{t-1}_{mw'}-\eta^t\hat{L}^{t-1}_{mw}+3=\ln\Big(\frac{X^t_{mw'}}{X^t_{mw}}\Big)+3, \ \forall t\ge 0.
\end{align}
For any $w'\neq w$, by taking limit from the above relation as $t\to \infty$, we get
$\limsup_{t\to \infty}\eta^{t}A^{t-1}\leq \ln(\frac{X^*_{mw'}}{X^*_{mw}})+3\leq O(1)$, where if $X^*_{mw'}=0$, by convention we define $\ln(\frac{X^*_{mw'}}{X^*_{mw}})=-\infty$. Now let $\kappa:=v_{mw'}(X^*)-v_{mw}(X^*)$. We claim $\kappa\leq 0$. Otherwise, if $\kappa>0$, then for any sufficiently large $t>\bar{t}$ we must have $A^{t-1}>t\kappa/2$. Given the choice of stepsize, we obtain $\lim_{t\to \infty}\eta^t A^{t-1}=\lim_{t\to \infty}t\eta^t\kappa/2=\infty$, which contradicts the fact that $\limsup_{t\to \infty}\eta^{t}A^{t-1}\leq O(1)$. Thus, $\kappa\leq 0$, which shows that if $X^*_{mw}>0$, then $v_{mw}(X^*)=\max_{w'}v_{mw'}(X^*)$, i.e., $X^*$ is a NE.\end{proof}

\section{Global Learning of Stable Matchings in General Markets}\label{sec:global}

In the previous section, we showed that the decentralized and uncoordinated EXP algorithm converges \emph{locally} to a stable matching at an exponential rate in general matching markets. This prompts the question of whether there is an uncoordinated and decentralized algorithm capable of \emph{globally} learning a stable matching in general markets, regardless of the convergence rate. In this section, we raiser this question affirmatively by introducing an alternative algorithm that leverages the weak acyclic property of the stable matching game. To that end, we first present several useful properties of the stable matching game, which will be instrumental in establishing our main convergence results.

\smallskip
\begin{definition} 
In a pure strategy profile $x=(x_m, m\in M)$, a matched man is the one whose pure strategy $x_m$ has exactly one coordinate equal to 1 (say $w$th coordinate), and all other coordinates are zero, and no other man of higher preference proposes to $w$, i.e., $x_{kw}=0\ \forall k>_w m.$ A pure strategy profile $x$ is a \emph{good state} if all the men matched under $x$ are at their best responses. In other words, a good state $x$ refers to a partial matching in which the matched men are at their best responses.
\end{definition} 

\smallskip
The following theorem shows that if the stable matching game with known preferences starts from a good state, then any sequence of best response dynamics by the players reaches a pure NE. Unless stated otherwise, we assume that each player $m$ is aware of both his own true preferences and those of all the women, enabling him to compute his payoff function. This assumption will be relaxed in Theorem \ref{thm:global-convergence-weakly}, where players seek to learn these unknown preferences through interactions with others. 

\smallskip
\begin{lemma}\label{the:potential}
Consider the stable matching game that starts from a good initial state $x^0$. Consider any sequence of updates where at each time $t$, an arbitrary subset of players that includes at least one unsatisfied player updates their actions by playing their best responses. Then, the sequence of updates converges to a pure NE in no more than $|M||W|$ iterations.
\end{lemma}
\begin{proof}
Let us encode the ordinal preferences of each woman $w$ into cardinal values $\lambda_{wm}\in \mathbb{R}^+$ such that $m_1>_w m_2$ if and only if $\lambda_{wm_1}>\lambda_{wm_2}$. Given a  pure strategy profile $x$, let $r_w(x)\in \{\lambda_{wm}, m\in M\}$ be woman $w$'s reward for her most preferred man who is matched to her under $x$. By convention, we set $r_w(x)=0$ if $w$ is not matched to any man under the pure-strategy profile $x$. We show that $r_w(x)$ is nondecreasing for each woman $w$ (and for at least one $w$ it is strictly increasing) such that the potential function $\Phi(x)=\sum_wr_w(x)$ strictly increases after each iteration. Since each $r_w(x)$ can increase at most $|M|$ times, the sequence must terminate at a pure NE after no more than $|M||W|$ iterations. We show this by arguing that during the sequence of updates, no matched man $m$ wants to deviate unless he is rejected by his partner $w$ because she has received a better offer, implying that $r_w(x)$ must have increased. Note that this statement holds initially because $x^0$ is a good state where no matched man wants to deviate.    

More precisely, let us by contradiction assume that $t$ is the first iteration at which some woman $w$ loses her more preferred man during the sequence of updates, i.e., $r_w(x^t)<r_w(x^{t-1})$. This can happen only if man $m$ that is matched to $w$ at time $t-1$ (i.e., $x^{t-1}_{mw}=1$) is selected to update at time $t$ whose best response is to leave $w$ for another women $w'$ (i.e., $x^t_{mw'}=1$). Otherwise, $r_w(x^{t-1})$ can only increase or remain the same because women reject any proposal other than their most preferred one. Note that since the game starts from a good state $x^0$, we must have $t\ge 2$.  

According to the payoff structures, the best response of player $m$ at time $t$ is to propose to his most preferred \emph{achievable} woman $w'$, i.e., among all women who rank $m$ higher or equal to their match under $x^{t-1}$, $m$ proposes to the woman $w'$ which brings him the highest utility $\mu_{mw'}$. Note that in particular $\mu_{mw'}>\mu_{mw}$. Thus, the first time $\tau\leq t-1$ when man $m$ proposed and got matched to woman $w$, the woman $w'$ was his more preferred choice and the reason why $m$ did not propose to $w'$ at time $\tau$ was because $w'$ was not achievable, i.e., $x^{\tau-1}_{m'w'}=1$ for some man $m'>_{w'} m$. Since $w'$ was not achievable for $m$ in $x^{\tau-1}$ but she is achievable in $x^{t-1}$ implies that $r_{w'}(x^{t-1})\!<\!r_{w'}(x^{\tau-1})$. Therefore, $r_{w'}(x^{\tau'})\!<\!r_{w'}(x^{\tau'-1})$ for some $\tau\leq \tau'\leq t-1$, contradicting the choice of $t$.\end{proof}

\smallskip
In fact, one can view the DA algorithm of \cite{gale1962college} as a special version of the above lemma in which all the players are selected to best respond at each round, and the initial good state is the empty matching. Lemma \ref{the:potential} shows that the stable matching game admits an ordinal potential function if the game starts from a good state (we refer to Appendix \ref{appx:potential} for another interpretation of this potential function). However, that does not imply that the stable matching game is an \emph{ordinal potential game} \cite{monderer1996potential}, which requires the existence of a potential function whose changes are aligned with players' payoff changes from any initial state. Nonetheless, the existence of such a ``weak'' potential function allows us to show the weak acyclicity of the stable matching game as defined next.

\smallskip
\begin{definition}[\cite{marden2007regret}]
A better response path in a finite action noncoopertaive game is a sequence of pure strategy profiles $x^1,x^2,\ldots,x^{L}$ such that for each $\ell=1,\ldots,L-1$, $x^{\ell+1}$ is obtained from $x^{\ell}$ by letting some player $i_{\ell}$ to play a better response. A game is called \emph{weakly acyclic} if from any pure strategy profile $x^0$, there is a finite length better response path to a pure NE.
\end{definition}

\smallskip
\begin{definition}
A pure NE $x^*=(x^*_m,x^*_{-m})$ of the stable matching game is called \emph{strict} if each player has a unique best response at the equilibrium $x^*$, i.e., for every $m$ and any pure strategy $x_m\neq x^*_m$, we have $u_m(x_m,x^*_{-m})<u_m(x^*_m,x^*_{-m})$.   
\end{definition}

In a weakly acyclic game, an arbitrarily high probability convergence can be guaranteed to a strict NE under appropriate learning dynamics \cite{marden2009payoff}. If we do not have strictness property, the learning dynamics may oscillate between two nonstrict NE points with a positive probability. Using the above definitions, we are now ready to state the main result of this section, which provides a decentralized and uncoordinated algorithm for learning stable matchings in general markets with unknown preferences. The algorithm is an adaptation of the sample experimentation algorithm for weakly acyclic games \cite{marden2009payoff} to the stable matching game. Intuitively, each player $m$ goes through a sequence of phase-dependent exploration/exploitation. At the end of each episode, player $m$ evaluates his payoff for each action he has taken during the last episode. He then updates his baseline action (an action that brings him the highest reward in the past episode up to some tolerance) and sticks to that baseline action most of the time in the next episode while still exploring new actions with a small probability. The detailed description of the algorithm is summarized in Algorithm \ref{alg:weak-acyclic}.   

\begin{algorithm}[t]\caption{A Globally Convergent Decentralized and Uncoordinated Algorithm for Player $m$}\label{alg:weak-acyclic}
{\bf Input:} A randomly chosen initial action $\alpha^0_m\in [W]$, an initial baseline action $b^1_m=\alpha^0_m$, episode length $\tau_m$, tolerance level $\delta >0$, exploration probability $\epsilon\in (0,1)$, and an inertia probability $\omega\in (0, 1)$.

\noindent
{\bf For} $s=1,2,\ldots$, do the following
\begin{itemize} 
\item During the $s$th episode, player $m$ selects his baseline action obtained at the end of the previous episode with probability $1-\epsilon$ or explores a new uniformly sampled action with probability $\epsilon$: 
\begin{align}\nonumber
\alpha^t_m=\begin{cases}
b^{s}_m &\mbox{w.p.}\ 1-\epsilon,  \\
w\in \mbox{Unif}[W] & \mbox{w.p.}\ \frac{\epsilon}{n}, 
\end{cases}\ \ \ \ \ \ \ \ \ \ \ \ \forall t\in \{(s-1)\tau_m,\ldots,s\tau_m-1\}.
\end{align}
\item  At the end of episode $s$, player $m$ evaluates his average utility when playing each action as
\begin{align}\nonumber
u^s_{mw}=\frac{1}{n^s_{mw}}\sum_{t=(s-1)\tau_m}^{s\tau_m-1} \mu^t_{mw}\boldsymbol{1}_{\{\alpha^t_k\neq w\ \forall k>_w m\}}\boldsymbol{1}_{\{\alpha_m^t=w\}}\ \ \ \ \forall w,
\end{align}
where $n^s_{mw}=\sum_{t=(s-1)\tau_m}^{s\tau_m-1}\boldsymbol{1}_{\{\alpha_m^t=w\}}$, and computes the set $\mathcal{A}^s_m=\{w: u^s_{mw}\ge u^s_{mb_m^s}+\delta\}$. 
\item If $\mathcal{A}^s_m\neq \emptyset$, player $m$ uniformly samples an action $w\in \mathcal{A}^s_m$, and updates his baseline action to $b_m^{s+1}=w$ with probability $1-\omega$. Otherwise, he sets $b_m^{s+1}=b^{s}_m$.
\end{itemize}
\end{algorithm}


\smallskip
\begin{theorem}\label{thm:global-convergence-weakly}
Let $p,\omega\in (0,1)$ be arbitrary probabilities and $\epsilon\leq \min\{\frac{1-p}{n}, \frac{\delta}{4n}, \frac{\Delta-\delta}{4n}\}$, where $\delta\in (0, \Delta)$. For each man $m$, there exists a sufficiently large episode length $\tau_m=\tau_m(p,\epsilon)$, such that if each player $m$ follows Algorithm \ref{alg:weak-acyclic} with episode length $\tau_m$, for all sufficiently large times $t>0$, the pure strategy profile $\alpha^t$ will be a (fixed) stable matching with probability at least $p$.
\end{theorem}   
\begin{proof}
First, we show that the stable matching game is a weakly acyclic game. If the game starts from a good state, by Lemma \ref{the:potential}, any sequence of best responses by the players is a better response path that leads to a pure NE in $|M||W|$ iterations. Thus, we only need to show that from any initial state, there is a better response path to a good state. This is also true by sequentially letting a matched man who wants to deviate (if there is such a man) to play his best response. In this way, either the number of matched men decreases due to collisions (which can happen at most $|M|$ times), or between every two consecutive collisions, there could be at most $|M|$ best response updates. Therefore, $|M|^2$ steps are enough to bring us from any arbitrary state to a good state, and the overall number of best response iterates to obtain a pure NE starting from any initial state is at most $|M||N|+|M|^2$.


Next, we note that any pure NE in the stable matching game is a strict NE. This follows directly by the fact that there are no ties between men's preferences in the stable matching game. More precisely, let $\alpha^*$ be a pure NE (stable matching) and assume $u_m(\alpha^*)=\mu_{mw}$. For any $\alpha_m\neq \alpha^*_m$, either $u_m(\alpha_m,\alpha^*_{-m})$ equals $0$ or $\mu_{mw'}$ for some $w'\neq w$. Because $\alpha^*_m$ is the best response of player $m$ at the NE and there are no ties between man $m$'s preferences, we must have $\mu_{mw'}<\mu_{mw}$. Since $0<\mu_{mw}$, in either case we have $u_m(\alpha_m,\alpha^*_{-m})<u_m(\alpha^*_m,\alpha^*_{-m})$.

Finally, it was shown in \cite{marden2009payoff} that the class of sample experimentation dynamics in weakly acyclic games converges with arbitrarily high probability to a strict NE if for any joint baseline action $b^s$ at the start of an episode $s$, $u^s_{mw}$ is an arbitrarily close estimation of $u_{m}(w,b^s_{-m})$ for sufficiently large episode length $\tau_m=\tau_m(p,\epsilon)$. Since Algorithm \ref{alg:weak-acyclic} is an adaptation of the sample experimentation dynamics to the stable matching game, the convergence result follows from \cite[Theorem 3.4]{marden2009payoff} if we can show that for any $\delta^*=\min\{\frac{\delta}{4}, \frac{\Delta-\delta}{4}\}>0$, $p<1$, and $1-(1-\epsilon)^{n-1}<\frac{\delta^*}{2}$, for sufficiently large $\tau_m$ we have
\footnote{This is an equivalent version of Claim 3.4 in \cite{marden2009payoff} to the stable matching game.}  
\begin{align}\label{eq:weak-law}
\mathbb{P}\Big\{\big|u_{mw}^s-u_{m}(w,b^s_{-m})\big|\ge \delta^*\Big\}\leq 1-p.
\end{align}
This statement is also true because during each episode, players sample their baseline actions (or explore a new action) independently of others, and moreover, $\mu^t_{mw}$ are i.i.d. realizations of a random variable with mean $\mu_{mw}$. Thus, relation \eqref{eq:weak-law} holds using the law of large numbers, which completes the proof.\end{proof}

\smallskip
While Algorithm \ref{alg:weak-acyclic} provides a decentralized and uncoordinated algorithm for learning stable matchings in general markets with arbitrarily high probability, the convergence time to a NE may take a long time. This seems inevitable to some extent in view of the exponential lower bound in terms of $n$ such that the best response dynamics in general matching markets converge to a stable matching \cite{ackermann2008uncoordinated}. 


\section{Conclusions}\label{sec:conclusion}

In this work, we considered the problem of learning stable matchings with unknown preferences in a fully decentralized and uncoordinated manner. Using a game-theoretic formulation, we showed how to leverage the rich literature for learning NE in noncooperative games to design principled learning algorithms for driving the matching market to its stable points. In particular, we established several global and local convergence results with performance guarantees for learning stable matchings in hierarchical markets, general markets, and under weak/strong information feedback. Our results provide new insights for learning stable matchings through NE learning in noncooperative games and bridge the discrete problem of learning a stable matching with learning NE in continuous-action games.  

This work opens several avenues for future research, some of which are discussed below.

\begin{itemize}
    \item We showed in Theorem \ref{thm:logarithmic-regret} that for hierarchical markets it is possible to learn a stable matching with logarithmic regret using a fully decentralized and uncoordinated algorithm. However, the regret bound includes a constant factor that scales exponentially with the size of the market $n$. While we believe this exponential factor is necessary when using the EXP algorithm, an interesting question is whether there exist other fully decentralized and uncoordinated algorithms that can eliminate this exponential factor from the regret bound. More generally, it is important to characterize the information-theoretic limits above and below which efficient learning of a stable matching in a fully decentralized and uncoordinated fashion may or may not be possible.    
    \item In Theorem \ref{thm:positive-prob-mixed-NE}, we showed that if the EXP algorithm converges with positive probability, it must converge to a stable matching. We conjecture that, for general markets, Algorithm \ref{alg:main} converges globally to a stable matching with arbitrarily high probability. In that case, Algorithm \ref{alg:main} provides a simpler decentralized and uncoordinated alternative to Algorithm \ref{alg:weak-acyclic} that globally converges to a stable matching in general matching markets.  
    \item The waiting list feedback provided in Appendix \ref{appx:monotone-game} is merely an initial attempt to design faster decentralized and uncoordinated algorithms under richer information feedback. Extending our results to include other types of information feedback and studying their effect on the convergence speed of the market to its stable points would be another interesting research direction.
\end{itemize}


%


\bibliographystyle{IEEEtran}
\bibliography{thesisrefs}

\appendices

\section{Matching Markets with Stronger Information Feedback}\label{appx:monotone-game}

So far in our work, we have assumed that men receive the minimum amount of information feedback as they interact with the market; namely, men only get to observe whether their proposal gets rejected or, if matched, they observe a random realization of their preference. In this appendix, we consider the learning problem from a new perspective by addressing the following question: \emph{Is there stronger information feedback that will allow us to derive the market faster globally to a stable matching in a decentralized and uncoordinated fashion?}  Therefore, we address the problem of learning 
a stable matching from an information feedback design perspective by showing that if the players are allowed to receive a bit more information as feedback, then the stable matching game will become a \emph{monotone game}, hence allowing a wide range of decentralized and uncoordinated learning algorithms for obtaining its NE points.

Let us consider a simplified version of the stable matching game with the same set of players $M$ and strategy spaces $\mathcal{X}_m, m\in M$, except that the payoff of player $m$ is now given by   
\begin{align}\label{eq:simplified}
u_m(x)=\sum_{w\in W}\mu_{mw}\big(1-\sum_{k>_w m}x_{kw}\big)x_{mw}.
\end{align}

The simplified payoff function \eqref{eq:simplified} can be viewed as the expected perceived reward for player $m$ if he randomly and independently of others proposes to women according to his mixed strategy $x_m\in \mathcal{X}_m$ but under slightly stronger feedback: if player $m$'s proposal gets rejected from a woman $w$, he also gets to know how many men higher ranked than him proposed to $w$. In this way, each man $m$ learns how many more preferred men compete with him to get matched to woman $w$. Such information feedback can be implemented by using a waiting list order for each woman, which shows the ranking of a rejected man $m$ in that list. Henceforth, we refer to such information feedback as the \emph{waiting list} feedback. 

Let $\alpha_k\in W$ denote the random woman that player $k$ proposes to her, which is drawn independently from his mixed strategy distribution $x_k$. Under the waiting list feedback information, player $m$ has access to the information $\sum_{k>_w m}\boldsymbol{1}_{\{\alpha_k=w\}}$ if he proposes to $w$, and hence can estimate his perceived payoff by proposing to a woman $w$ as $\mu_{mw}(1-\sum_{k>_w m}\boldsymbol{1}_{\{\alpha_k=w\}})$,  Therefore, the expected perceived payoff of player $m$ under the waiting list feedback equals
\begin{align}\nonumber
\mathbb{E}[\sum_{w\in W}\mu_{mw}(1-\sum_{k>_w m}\boldsymbol{1}_{\{\alpha_k=w\}})\boldsymbol{1}_{\{\alpha_m=w\}}]&=\sum_{w\in W}\mathbb{E}[\mu_{mw}(1-\sum_{k>_w m}\boldsymbol{1}_{\{\alpha_k=w\}})|\alpha_m=w]\mathbb{P}(\alpha_m=w)\cr
&=\sum_{w\in W}\mu_{mw}\big(1-\sum_{k>_w m}x_{kw}\big)x_{mw}=u_m(x),
\end{align}    
where the second equality holds by independency of $\alpha_k\ \forall k\in M$. 

The motivation for waiting list feedback and the payoff function \eqref{eq:simplified} is that if we view $1-\sum_{k>_w m}x_{kw}$ as the remaining fractional capacity of women $w$ for being matched, using the approximation $1-\sum_{k>_w m}x_{kw}\sim e^{-\sum_{k>_w m}x_{kw}},$ one can think that man $m$ loses his interest to offer women $w$ exponentially in terms of the occupied capacity of woman $w$ by higher preferred men. One immediate consequence of such simplified payoff functions is the following proposition whose proof follows nearly verbatim arguments as those in Theorems \ref{thm:NE-Stable} and \ref{thm:mixed-pure}, and is omitted from here for the sake of brevity.

\begin{proposition}\label{prop:list}
Consider the stable matching game with payoff functions \eqref{eq:simplified}. Then, $x^*$ is a pure NE if and only if it is the characteristic vector of a stable matching. Moreover, any mixed NE of this game with full support can be rounded in a decentralized manner to a pure NE.  
\end{proposition}

Next, we consider the following definition of monotone games \cite{rosen1965existence}. 
\begin{definition}
A continuous-action game $(M,\{U_m\}_{m\in M}, \{\mathcal{X}\}_{m\in M})$ is called monotone if for any two action profiles $x,x'\in \mathcal{X}_1\times\cdots\times \mathcal{X}_{|M|}$, we have $(x'-x)^T\big(F(x')-F(x)\big)\leq 0$, where $$F(x)=(\nabla_m U_m(x), m\in M)^T.$$  
\end{definition}

An important property of monotone games is that they admit a variety of decentralized learning algorithms that converge globally to a NE, which, under further assumptions, it can be shown that the converge rate could be exponential or polynomial \cite{tatarenko2024payoff,rosen1965existence,hsieh2021adaptive,gao2020continuous,tatarenko2022rate,mertikopoulos2019learning}. The following theorem shows that the \emph{regularized} stable matching game with waiting list feedback is indeed a monotone game under a suitable assumption on the preferences. 

\begin{theorem}\label{lemm:list-monotone}
Consider the regularized stable matching game with list feedback, where each player aims to maximize its regularized payoff function given by 
\begin{align}\nonumber
U_m(x_m,x_{-m}):=u_m(x_m,x_{-m})-\frac{\beta}{2}\|x_{m}\|^2.
\end{align}
Then, for any two strategy profiles $x',x$, we have
\begin{align}\nonumber
(x'-x)^T\big(F(x')-F(x)\big)=-\frac{1}{2}\sum_{w\in W}(x_{w}'-x_{w})^TQ^w(x_{w}'-x_{w}),
\end{align} 
where $x_w$ is the vector $x_w=(x_{mw}, m\in M)^T$, and $Q^w$ is a symmetric matrix whose diagonal entries are $2\beta$ and all the entries in row $m$ up to the diagonal entry in that row equal $\mu_{mw}$.  In particular, if $\beta>\frac{n\mu_{\max}}{2}$, the regularized stable matching game with waiting list feedback is a monotone game.  
\end{theorem}
\begin{proof}
Note that due to the linearity of $U_m(x)$ with respect to player $m$'s strategy $x_m$, we have 
\begin{align}\nonumber
\nabla_{mw} U_m(x)=\mu_{mw}(1-\sum_{k>_w m}x_{kw})-\beta x_{mw},\  \forall m,w,
\end{align}
where $\nabla_{mw} U_m(x)$ denotes the gradient of $U_{m}(x)$ with respect to $x_{mw}$. Thus, we can write
\begin{align}\label{eq:monotone-first}
(x'&-x)^T\big(F(x')-F(x)\big)=\sum_m (x_m'-x_m)^T \big(\nabla_m U_m(x')-\nabla_m U_m(x)\big)\cr 
&=\sum_m\sum_w (x_{mw}'-x_{mw})\big(\nabla_{mw} U_m(x')-\nabla_{mw} U_m(x)\big)\cr 
&=\sum_m\sum_w -\mu_{mw}(x_{mw}'-x_{mw})\big(\sum_{k>_w m}x'_{kw}-\sum_{k>_w m}x_{kw}+\beta(x'_{mw}-x_{mw})\big)\cr
&=-\sum_m\sum_w\sum_{k>_w m} \mu_{mw}(x_{mw}'-x_{mw})(x'_{kw}-x_{kw})-\beta\sum_m\sum_w (x'_{mw}-x_{mw})^2\cr
&=-\sum_w\Big(\sum_m\sum_{k>_w m} \mu_{mw}(x_{mw}'-x_{mw})(x'_{kw}-x_{kw})\Big)-\beta\sum_m\sum_w (x'_{mw}-x_{mw})^2\cr 
&=-\sum_w\Big(\sum_m\sum_{k} \mu_{mw}\boldsymbol{1}_{\{k>_w m\}}(x_{mw}'-x_{mw})(x'_{kw}-x_{kw})\Big)-\beta\sum_w\sum_m (x'_{mw}-x_{mw})^2.
\end{align}
For each woman $w$, let us define a vector $x_w=(x_{mw}, m\in [M])^T$, and an $n\times n$ incidence preference matrix $A^{w}$ whose entries are given by $A^w_{mk}=\mu_{mw}\boldsymbol{1}_{\{k>_w m\}}, \forall m,k\in [n]$. Then, using \eqref{eq:monotone-first} we have
\begin{align}\nonumber
(x'-x)^T\big(F(x')-F(x)\big)&=-\sum_w\Big(\sum_m\sum_{k} A^w_{mk}(x_{mw}'-x_{mw})(x'_{kw}-x_{kw})\Big)\cr
&\qquad-\sum_w (x'_{w}-x_w)^T(\beta I)(x'_{w}-x_w)\cr 
&=-\sum_w(x_{w}'-x_{w})^T(A^w+\beta I) (x_{w}'-x_{w})\cr
&=-\frac{1}{2}\sum_w(x_{w}'-x_{w})^TQ^w(x_{w}'-x_{w}),
\end{align}
where $Q^w=(A^w)^T+A^w+2\beta I$. Next, we note that permuting the rows and columns of a matrix by the same permutation matrix preserves the semi-definiteness property as it only changes the basis of its eigenvectors. Thus, without loss of generality, we may assume $1>_w 2>_w \ldots>_w n$. However, in that case, we may assume that for any $w$, the matrix $A^{w}$ will be a lower triangular matrix whose diagonal entries are all zero. In particular, $Q^w=(A^w)^T+A^w+2\beta I$ would be a symmetric matrix whose elements in each row $m$ up to the diagonal entry are identical and equal to $\mu_{mw}$. Finally, for sufficiently large $\beta>\frac{n\mu_{\max}}{2}$, the matrix $Q^w$ is positive semidefinite for all $w$, and we have $-(x_{w}'-x_{w})^TQ^w(x_{w}'-x_{w})\leq 0 \ \forall w$. This shows that $(x'-x)^T\big(F(x')-F(x)\big)\leq 0$, and completes the proof.
\end{proof}

The monotonicity property in Theorem \ref{lemm:list-monotone} holds when the regularizer parameter $\beta$ is sufficiently large. The regularizer implies that players tend to maximize their approximate payoff functions. Thus, instead of learning a NE for the payoff functions \eqref{eq:simplified}, they learn an $\epsilon$-NE where the approximation error $\epsilon$ grows as $\beta$ increases. As a result, applying decentralized learning algorithms to learn a NE of the regularized game only guarantees an $\epsilon$-stable matching in the sense that the players will achieve a solution in which each player would be $\epsilon$ worse off than what he could achieve in a stable matching. Therefore, the regularizer parameter trades off between the market's convergence speed and the accuracy of the learned stable matching. Finally, we note that the assumption of waiting list feedback is not necessary for faster convergence in any general market; it is merely one choice that we consider in this work, and there could be other types of information feedback that allow faster convergence. Nevertheless, our game-theoretic approach to the stable matching problem provides valuable insights on how to properly design the utilities and the information feedback to devise faster decentralized and uncoordinated algorithms for learning stable matchings.

\section{Supplementary Materials and Omitted Proofs}\label{appx:omitted}
 
\medskip
\subsection{{\bf Freedman's Concentration 
Inequality for Martingales}}\label{appx:freedman} 
\begin{lemma}\label{lemm:freedman}
Let $\{\zeta_{\tau}\}_{\tau=1}^t$ be a martingale difference sequence with $\max_{\tau\in [t]}|\zeta_{\tau}|\leq u$, and consider the corresponding martingale $S^t=\sum_{\tau=1}^{t}\zeta_{\tau}$ with $S^1=0$. Let $V^t=\sum_{\tau=1}^{t}\mathbb{E}[\zeta^2_{\tau}|\mathcal{F}^{t-1}]$ be the sum of the conditional variances. Then, for any $x>0$ and $v\in \mathbb{R}$, we have $$\mathbb{P}\{|S^t|\ge x,\ V^t\leq v^2\}\leq 2\exp\Big(\frac{-x^2}{2(v^2+ux)}\Big).$$
\end{lemma}
\begin{proof}
The proof can be found in \cite[Theorem 1.6]{freedman1975tail}.\end{proof}


\medskip
\subsection{{\bf Proof of Lemma \ref{lemm:local_neighbor}}}\label{sec:local_neighbor}

\begin{proof}
Fix a player $m$ and assume $X^*_{mw^*}=1$. Since $X^*$ is a pure NE, for any $w\neq w^*$ we have
\begin{align}\nonumber
v_{mw^*}(X^*)=\mu_{mw^*}> v_{mw}(X^*)=\mu_{mw}\prod_{k>_w m}(1-X^*_{kw})\in \{0, \mu_{mw}\}. 
\end{align}
On the other hand, we have 
\begin{align}\label{eq:m-w-w}
v_{mw^*}(x)-v_{mw^*}(X^*)&=\mu_{mw^*}\prod_{k>_{w^*} m}(1-x_{kw^*})-\mu_{mw^*}\cr
&\ge -\mu_{mw^*}\sum_{k>_{w^*} m}x_{kw^*}\cr
&\ge -\frac{\mu_{mw^*}}{2}\|x-X^*\|_1,
\end{align}
where the second inequality holds because $X^*_{mw^*}=1$ implies $X^*_{kw^*}=0\ \forall k\neq m$, and hence $\frac{1}{2}\|x-X^*\|_1\ge \sum_{k>_{w^*} m}x_{kw^*}$. Now we consider two cases: 

i) If $v_{mw}(X^*)=\mu_{mw}$, then $v_{mw}(x)\leq \mu_{mw}< \mu_{mw^*}\ \forall x$, which implies $v_{mw^*}(X^*)-v_{mw}(x)\ge \Delta$. Combining this relation with \eqref{eq:m-w-w} and because $\|x-X^*\|_1\leq \frac{\Delta}{\mu_{\max}}$, we get $v_{mw^*}(x)-v_{mw}(x)>\frac{\Delta}{2}>c$. 

ii) If $v_{mw}(X^*)=0$, there exists $k>_w m$ such that $X^*_{kw}=1$. Then, $v_{mw}(x)\leq \mu_{mw}(1-x_{kw})= \mu_{mw}(X^*_{kw}-x_{kw}) \leq \frac{1}{2}\mu_{mw}\|X^*-x\|_1$, and we have $v_{mw^*}(X^*)-v_{mw}(x)\ge \mu_{mw^*}-\frac{1}{2}\mu_{mw}\|X^*-x\|_1$. Combining this relation with \eqref{eq:m-w-w} and because $\|x-X^*\|_1\leq \frac{\mu_{\min}}{\mu_{\max}}$, we get $v_{mw^*}(x)-v_{mw}(x)>\frac{\mu_{\min}}{2}> c.$
\end{proof} 

\medskip
\subsection{{\bf Proof of Lemma \ref{lemm:tech-bound}}}\label{sec:proof-lemma}

\begin{proof}
i) Suppose $\|X^t-X^*\|_1\leq \frac{c}{250n^2}$, then $\|X^t_m-X^*_m\|_1=2(1-X^t_{mw^*})\leq \frac{c}{250n^2}\ \forall m$. Since $X^t_{m}$ is a probability distribution, we can write $$\frac{X^t_{mw^*}}{X^t_{mw}}\ge \frac{1-c/500n^2}{c/500n^2}= \frac{500n^2}{c}-1, \forall m,w\neq w^*.$$ 
Therefore, for any $m,w\neq w^*$, we have $\eta^{t}\hat{L}^{t-1}_{mw^*}-\eta^{t}\hat{L}^{t-1}_{mw}\ge \ln (\frac{500n^2}{c}-1)\ge \ln (\frac{n^2}{2c})+6$.

\noindent
ii) Suppose $\eta^{t}\hat{L}^{t-1}_{mw^*}-\eta^{t}\hat{L}^{t-1}_{mw}\ge \ln(\frac{n^2}{2c})\ \forall m,w$. Then, we can write
\begin{align}\nonumber
\frac{X^t_{mw^*}}{X^t_{mw}}=\exp\big(\eta^{t}\hat{L}^{t-1}_{mw^*}-\eta^{t}\hat{L}^{t-1}_{mw}\big)\ge \frac{n^2}{2c} \ \forall m,w\neq w^*, 
\end{align}
which implies $X^t_{mw^*}\ge \frac{1}{1+2c/n}\ge 1-\frac{2c}{n}$. Thus, $\|X^t_m-X^*_m\|_1=2(1-X^t_{mw^*})\leq \frac{4c}{n}\ \forall m$, and hence $\|X^t-X^*\|_1\leq c$.

\noindent
iii) Suppose $\eta^{t}A^{t-1}=b+\beta$, where $b=\ln(\frac{n^2}{2c})+3$ and $\beta\ge 0$. Then, we have
\begin{align}\nonumber
\eta^{t+1}A^{t}&=\eta^{t+1}A^{t-1}+\eta^{t+1}(v_{mw^*}(X^{t})-v_{mw}(X^{t}))\cr
&\ge \eta^{t+1}A^{t-1}+\eta^{t+1} c\cr
&= \frac{\eta^{t+1}}{\eta^{t}}(b+\beta)+\eta^{t+1} c\cr
&=b+\frac{\eta^{t+1}}{\eta^{t}}\beta+\big(\eta^{t+1} c-(1-\frac{\eta^{t+1}}{\eta^{t}})b\big)\cr 
&\ge b+\eta^{t+1}\big(c-(\frac{1}{\eta^{t+1}}-\frac{1}{\eta^{t}})b\big)\cr
&\ge b,
\end{align} 
where the last inequality holds by the fact that $\frac{1}{\eta^{t+1}}-\frac{1}{\eta^{t}}\leq \frac{c}{b}\ \forall t\ge t_0$.
\end{proof}


\medskip
\subsection{{\bf An Alternative Form for the Potential Function in Lemma \ref{the:potential}}}\label{appx:potential}
One can represent the potential function in Lemma \ref{the:potential} directly for womens' \emph{ordinal} preferences. Given a mixed strategy profile $x$, let us associate a vector $p_w(x)\in \mathbb{R}_+^{|M|}$ to each woman $w$ whose coordinates are sorted in the order of $w$'s preferences, such that if $m_1>_wm_2>_w\ldots>_w m_n$, we let the $j$th coordinate of $p_w(x)$ be given by $p_{jw}(x)=x_{m_jw}\prod_{k>_w m_j}(1-x_{kw}),$ which is the probability that woman $w$ is matched to her $j$th favorite man, assuming that the players propose to women independently and according to their mixed strategies $x$. We note that for a pure-strategy profile $x\in \{0,1\}^{|M||W|}$, $p_w(x)$ would be a binary vector with all entries being zero except the $j$th one being one, where $j$ is the index of the highest ranked man proposing to $w$. Therefore, Lemma \ref{the:potential} asserts that after each best response move, the vectors $\{p_w(x), w\in W\}$ increase lexicographically with at least one of the vectors increases strictly. Now, if we again encode woman $w$'s preferences in a cardinal vector $\lambda_w=(\lambda_{wm_1},\ldots,\lambda_{wm_n})$, then for a pure strategy profile $x\in \{0,1\}^{|M||W|}$, we have $$r_w(x)=\langle p_w(x),\lambda_w \rangle=\sum_{m}\lambda_{wm}x_{mw}\prod_{k>_w m}(1-x_{kw}),$$ which coincides with the definition of $r_w(x)$ in Lemma \ref{the:potential} that picks the reward of the highest-ranked man proposing to $w$ under $x$. Therefore, in the space of cardinal preferences and for any pure-strategy profile $x\in \{0,1\}^{|M||W|}$, one can rewrite the potential function in Theorem \ref{the:potential} in a compact as 
\begin{align}\label{eq:potential-function}
\Phi(x)=\sum_w r_w(x)=\sum_{m,w}\lambda_{wm}x_{mw}\prod_{k>_w m}(1-x_{kw}),   
\end{align}
whose value increases after every best response update by a player. We note the similarity between the potential function \eqref{eq:potential-function} and sum of players' payoffs $\sum_{m}u_m(x)=\sum_{m,w}\mu_{mw}x_{mw}\prod_{k>_w m}(1-x_{kw})$, where one is written using mens' preferences $\{\mu_{mw}\}$ and the other using  womens' preferences $\{\lambda_{wm}\}$.

\medskip
\subsection{{\bf A Gap Estimation Oracle for Unknown $\Delta$.}}\label{appx:gap-oracle}

Motivated by the Upper Confidence Bound (UCB) for stochastic bandits \cite{seldin2017improved}, here we provide a gap estimation oracle, which, in combination with Algorithm \ref{alg:main}, allows each player to adaptively adjust his mixing parameter based on the estimated gap he has obtained so far. The integration of this oracle with Algorithm \ref{alg:main} allows the algorithm to be more practical to the situation when a lower bound for the gap $\Delta$ is unknown. In particular, using standard concentration bounds for the UCB algorithm, one can show that the event that the estimated gap is larger than the true gap decays fast as time increases, hence only introducing a small error term to the overall performance of the algorithm.

\begin{algorithm}[H]\caption{A Gap Estimation Oracle for Player $m$}\label{alg:gap-estimation}
\noindent
{\bf For} $t=1,2,\ldots$, player $m$ independetly performas the following steps:

\begin{itemize}
    \item player $m$ independently computes his estimated preferences $\bar{\mu}^{t-1}_{mw}$ as well as his upper and lower confidence bounds as
\begin{align}\nonumber
\bar{\mu}^{t-1}_{mw}=\frac{M^{t-1}_{mw}}{N^{t-1}_{mw}}, \ \ \ \ 
U^t_{mw}=\min\{1, \bar{\mu}^{t-1}_{mw}+\sqrt{\frac{2\log t}{N^{t-1}_{mw}}}\}, \ \ \ \ L^t_{mw}=\max\{0, \bar{\mu}^{t-1}_{mw}-\sqrt{\frac{2\log t}{N^{t-1}_{mw}}}\},
\end{align}
where $N^{t-1}_{mw}$ denotes the number of times that $m$ successfully gets matched to $w$ up to time $t-1$, and $M^{t-1}_{mw}$ denote the accumulated rewards by player $m$ of proposing to $w$ up to time $t-1$.
\item At time $t$, player $m$ proposes to a woman $\alpha^t_m$ according to Algorithm \ref{alg:main} using the mixing parameter $\gamma^t=\hat{M}\frac{\log t}{t}$, where $\hat{M}$ is the value of $M$ in which $\Delta$ and $\mu_{mw}$ are replaced by the estimated values $\hat{\Delta}^t_m=\min_w \hat{\Delta}^t_{mw}$, and $\bar{\mu}^{t-1}_{mw}$, where 
\begin{align}\nonumber
\hat{\Delta}^t_{mw}=\max\{0, L^t_{mw}-\min_{w'}U^t_{mw'}\}.   
\end{align}
\item Player $m$ updates its parameters as
\begin{align}\nonumber
&N^t_{mw}=N^{t-1}_{mw}+\boldsymbol{1}_{\{\alpha^t_k\neq w\ \forall k>_{w} m,\ \alpha_m^t=w\}},\cr
&M^t_{mw}=M^{t-1}_{mw}+\mu^t_{mw}\boldsymbol{1}_{\{\alpha^t_k\neq w\ \forall k>_{w} m,\ \alpha_m^t=w\}}. 
\end{align}
\end{itemize} 
\end{algorithm}

\end{document}